\documentclass[a4paper,11pt]{article}

\usepackage[a4paper,margin=1in]{geometry}
\usepackage[utf8]{inputenc}
\usepackage{amsmath}
\usepackage{amssymb}
\usepackage{amsthm}
\usepackage{booktabs}
\usepackage[small]{caption}
\usepackage{cite}
\usepackage{colortbl}
\usepackage{enumitem}
\usepackage{framed}
\usepackage{graphicx}
\usepackage{microtype}
\usepackage{xcolor}
\usepackage[unicode]{hyperref}
\usepackage{amsfonts}
\usepackage{algorithm}
\usepackage{algpseudocode}
\usepackage{bm}
\usepackage{mathtools}
\usepackage{subcaption}
\RequirePackage{fix-cm}

\usepackage{thm-restate}
\usepackage{cite}

\newcommand{\fakeparagraph}[2]{\par\noindent\textbf{#1}\hspace{1em}#2}

\algblock{Input}{EndInput}
\algnotext{EndInput}
\newcommand{\Desc}[2]{\State \makebox[2em][l]{#1}#2}
\algnewcommand\And{\textbf{and} }
\algnewcommand\Or{\textbf{or} }

\newfloat{problemdef}{htbp}{loa}
\floatname{problemdef}{Problem}
\newcommand{\problemcaptionkludge}{\rule[-.3\baselineskip]{0pt}{\baselineskip}}

\theoremstyle{plain}
\newtheorem{theorem}{Theorem}[section]
\newtheorem{lemma}[theorem]{Lemma}
\newtheorem{corollary}[theorem]{Corollary}
\newtheorem{observation}[theorem]{Observation}
\theoremstyle{definition}
\newtheorem{definition}[theorem]{Definition}

\theoremstyle{remark}

\newcommand{\neigh}{\mathcal{N}}
\newcommand{\graph}{\mathcal{G}}
\newcommand{\mender}{\mathcal{M}}

\newcommand{\tree}{\mathcal{T}}
\newcommand{\final}{\mathcal{F}}
\newcommand{\labels}{\Sigma}
\newcommand{\parts}{\mathcal{P}}
\newcommand{\NN}{\ensuremath{\mathbb{N}}}
\newcommand{\problem}{\Pi}

\DeclareMathOperator{\dom}{dom}
\DeclareMathOperator{\avg}{avg}
\DeclareMathOperator{\lcm}{lcm}
\DeclareMathOperator{\poly}{poly}
\DeclareMathOperator{\diff}{diff}

\newcolumntype{?}{!{\vrule width 1pt}}

\DeclareMathOperator{\minvolume}{\exists\textup{\textsf{MVol}}}
\DeclareMathOperator{\avgvolume}{\textup{\textsf{EMVol}}}
\DeclareMathOperator{\maxvolume}{\textup{\textsf{DMVol}}}
\DeclareMathOperator{\radius}{\textup{\textsf{MRad}}}
\DeclareMathOperator{\mendprob}{\textup{\textsf{Mend}}}
\DeclareMathOperator{\boldminvolume}{\exists\textbf{\textup{\textsf{MVol}}}}
\DeclareMathOperator{\boldavgvolume}{\textbf{\textup{\textsf{EMVol}}}}
\DeclareMathOperator{\boldmaxvolume}{\textbf{\textup{\textsf{DMVol}}}}
\DeclareMathOperator{\boldradius}{\textbf{\textup{\textsf{MRad}}}}

\newcommand{\radiustitle}{\texorpdfstring{\boldmath\(\radius\)}{MRad}}
\newcommand{\minvolumetitle}{\texorpdfstring{\boldmath\(\minvolume\)}{\textexists MVol}}
\newcommand{\avgvolumetitle}{\texorpdfstring{\boldmath\(\avgvolume\)}{EMVol}}
\newcommand{\maxvolumetitle}{\texorpdfstring{\boldmath\(\maxvolume\)}{DMVol}}
\newcommand{\neighsizetitle}{\texorpdfstring{\boldmath\(|\neigh_{\radius}|\)}{neighborhood size}}
\newcommand{\minvolumepolynomialtitle}{\texorpdfstring{\boldmath\(\minvolume = n^{\Theta(1)}\)}{\textexists MVol is polynomial}}
\newcommand{\minvolumepolylogtitle}{\texorpdfstring{\boldmath\(\minvolume = (\log n)^{\Theta(1)}\)}{\textexists MVol is polylogarithmic}}

\newcommand{\set}[1]{\ensuremath{\left\{#1\right\}}}

\hypersetup{
    colorlinks=true,
    linkcolor=black,
    citecolor=black,
    filecolor=black,
    urlcolor=[rgb]{0,0.1,0.5},
    pdftitle={Mending Partial Solutions With Few Changes},
    pdfauthor={Darya Melnyk, Jukka Suomela, Neven Villani}
}

\newcommand{\myaff}[1]{\,$\cdot$\, {\small #1}\par\smallskip}

\newenvironment{myabstract}
{\list{}{\listparindent 1.5em%
        \itemindent    \listparindent
        \leftmargin    0cm
        \rightmargin   0cm
        \parsep        0pt}%
    \item\relax}
{\endlist}

\newenvironment{mycover}
{\list{}{\listparindent 0pt
        \itemindent    \listparindent
        \leftmargin    0cm
        \rightmargin   0cm
        \parsep        0pt}%
    \raggedright
    \item\relax}
{\endlist}

\begin{document}

\begin{mycover}
{\fontsize{19.7pt}{19.5}\selectfont\bfseries Mending Partial Solutions with Few Changes \par}\hspace{-2pt}
\bigskip
\bigskip
\bigskip

\textbf{Darya Melnyk}
\myaff{Aalto University}

\textbf{Jukka Suomela}
\myaff{Aalto University}

\textbf{Neven Villani}
\myaff{Aalto University and École normale supérieure Paris-Saclay}

\end{mycover}
\medskip

\begin{myabstract}
\fakeparagraph{Abstract.}
In this paper, we study the notion of mending, i.e. given a partial solution to a graph problem, we investigate how much effort is needed to turn it into a proper solution. For example, if we have a partial coloring of a graph, how hard is it to turn it into a proper coloring?

In prior work (SIROCCO 2022), this question was formalized and studied from the perspective of \emph{mending radius}: if there is a hole that we need to patch, how \emph{far} do we need to modify the solution? In this work, we investigate a complementary notion of \emph{mending volume}: how \emph{many} nodes need to be modified to patch a hole?

We focus on the case of locally checkable labeling problems (LCLs) in trees, and show that already in this setting there are two infinite hierarchies of problems: for infinitely many values $0 < \alpha \le 1$, there is an LCL problem with mending volume $\Theta(n^\alpha)$, and for infinitely many values $k \ge 1$, there is an LCL problem with mending volume $\Theta(\log^k n)$. Hence the mendability of LCL problems on trees is a much more fine-grained question than what one would expect based on the mending radius alone.

We define three variants of the theme: (1)~\emph{existential} mending volume, i.e., how many nodes need to be modified, (2)~\emph{expected} mending volume, i.e., how many nodes we need to explore to find a patch if we use randomness, and (3)~\emph{deterministic} mending volume, i.e., how many nodes we need to explore if we use a deterministic algorithm. We show that all three notions are distinct from each other, and we analyze the landscape of the complexities of LCL problems for the respective models.
\end{myabstract}
\medskip

\section{Introduction}

If we have a partial solution to a graph problem, how much effort is needed to turn it into a proper solution? For example, if we have a partial coloring of a graph, how hard is it to turn it into a proper coloring? In this work we present three formalisms that capture the essence of this question; the first one is purely graph-theoretic while the other two are algorithmic:
\begin{enumerate}
    \item \textbf{Existential mending volume:} How many labels do we need to \emph{change} to ``patch a hole'' in the solution?
    \item \textbf{Expected mending volume:} In expectation, how many nodes do we need to \emph{explore} to learn enough about the input graph so that we can ``patch a hole''?
    \item \textbf{Deterministic mending volume:} In the worst case, how many nodes do we need to explore to learn enough about the input graph so that we can ``patch a hole''?
\end{enumerate}

We will define these concepts formally in Definition~\ref{def:existential_volume} and \ref{def:complexity_generalized}, but for now the following informal description will suffice to understand what we mean by ``patching a hole''. We are given a graph $G$, a partial solution $\lambda$ for some graph problem $\problem$, and some node $v$ that is unlabeled in $\lambda$. We would like to find a new solution $\lambda'$ such that node $v$ is labeled in $\lambda'$, and also all nodes that were already labeled in $\lambda$ remain labeled in $\lambda'$. We say that $\lambda'$ is a \emph{mend} of $\lambda$ at node $v$; we have ``patched a hole'' at $v$. Now the key complexity measure is the Hamming distance between $\lambda$ and $\lambda'$, i.e., the number of nodes that we had to change. If for any $G$, $\lambda$, and $v$ there is a mend $\lambda'$ that is within distance $T_1(n)$ from $\lambda$, we say that the existential mending volume of $\problem$ is at most $T_1(n)$. If there is a randomized algorithm that after exploring in expectation $T_2(n)$ nodes around $v$ can find a mend, we say that the expected mending volume is at most $T_2(n)$, and if there is a deterministic algorithm that after exploring in the worst case $T_3(n)$ nodes around $v$ can find a mend, we say that the deterministic mending volume is at most $T_3(n)$.

\subsection{Motivation}

Mending volume is intimately connected with the analysis of \emph{local search}. In particular, if the mending volume of problem $\problem$ is bounded by $T$, then we can start with any partial solution and walk towards a valid solution so that at each step we only need to consider modifications in which we change $T$ labels. 

Moreover, mending volume naturally captures the \emph{reconfiguration effort} in computer systems. The system is initially in a valid state, but the physical structure of the system changes (e.g., a new component is installed), leading to an invalid state $\lambda$ in which at least one component is unable to fulfill its task. We need to find a new configuration $\lambda'$ in which all components again function correctly. Further, in order to minimize service disruptions, we should also ensure that $\lambda'$ is as close to $\lambda$ as possible.

\subsection{Contributions}

It is easy to come up with graph problems where mending is trivial or very hard---these are problems with existential mending volume $O(1)$ or $\Theta(n)$. The work of Panconesi and Srinivasan~\cite{panconesi95delta} shows that the mending volume of $\Delta$-coloring in a graph of maximum degree $\Delta \ge 3$ is $O(\log n)$. But is the mending volume for problems of this flavor always $O(1)$, $\Theta(\log n)$, or $\Theta(n)$?

We formalize this question by considering \emph{locally checkable labeling problems} (LCLs), as defined by Naor and Stockmeyer~\cite{Naor1995}; these are problems in which we are given a graph with some maximum degree $\Delta$, and the task is to label the nodes with labels from some finite set $\labels$, subject to some local constraints. Graph coloring with $k = O(1)$ colors in a graph of maximum degree $\Delta = O(1)$ is a model example of an LCL problem.

We show that already in the case of trees, it is possible to construct two infinite hierarchies of problems: for infinitely many values $0 < \alpha \le 1$, there is an LCL problem with mending volume $\Theta(n^\alpha)$, and for infinitely many values $k \ge 1$, there is an LCL problem with mending volume $\Theta(\log^k n)$.

This shows that there is a striking difference between existential mending volume that we study here and the \emph{mending radius} that was defined recently in prior work~\cite{local_mending}. In trees, the mending radius of any LCL problem is known to be $O(1)$, $\Theta(\log n)$, or $\Theta(n)$. Hence mending volume makes it possible to classify LCL problems into infinitely many classes, while mending radius only leads to three classes of problems.

We also explore the landscape of mending volume beyond the case of trees; the results are summarized in Table~\ref{tab:complexities_overview}. In Section~\ref{sec:algorithmic_mending_volume}, we then further study the relation between existential, expected, and deterministic mending volumes. We show that there are LCL problems in which all three notions coincide, but that there are also problems that separate existential and randomized mending volumes, as well as problems that separate randomized and deterministic mending volumes. That is, in the worst-case partial solutions of some LCL problems, the most efficient mend can be well-hidden in the sense that it is hard to find by probing a graph. A summary of these results is presented in Table~\ref{tab:model_comparison}; we refer to Table~\ref{tab:complexities_complete} in Appendix~\ref{sec:landscape_appendix} for more details on the landscape of possible mending volumes.

\begin{table}[tb]
    \newcommand{\ok}{{\textcolor{blue!80!black}{\checkmark}}}
    \newcommand{\ko}{{\textcolor{red!80!black}{\textbf{\texttimes}}}}
    \newcommand{\uk}{{\textcolor{black}{?}}}
    \newcommand{\cs}{\qquad}
    \centering
    \caption{An overview of the landscape of existential mending volume ($\minvolume$) for LCL problems on the classes of paths, trees and general graphs. Here \ok{} denotes that LCL problems with this mending volume exist, \ko{} denotes that such LCL problems cannot exist, and \uk{} denotes an open question. See Table~\ref{tab:complexities_complete} in Appendix~\ref{sec:landscape_appendix} for the landscape of other notions of mending.}
    \label{tab:complexities_overview}
    \begin{tabular}{lc@{\cs}c@{\cs}c@{\cs}c@{\cs}c@{\cs}c@{\cs}c}
    \toprule
    Setting & \multicolumn{7}{l}{Possible mending volumes} \\ \cmidrule{2-8}
    & \(O(1)\) & \ldots & \(\Theta(\log n)\) & \(\Theta(\log^k n)\) & \ldots & \(\Theta(n^{\alpha})\) & \(\Theta(n)\) \\
    &&&& \makebox[0pt]{\footnotesize $k > 1$} && \makebox[0pt]{\footnotesize $0 < \alpha < 1$} \\
    \midrule
    Paths and cycles & \ok  & \ko  & \ko  & \ko  & \ko  & \ko  & \ok \\
    Rooted trees     & \ok  & \ko  & \ok  & \ok  & \ko  & \ok  & \ok \\
    Trees            & \ok  & \ko  & \ok  & \ok  & \uk  & \ok  & \ok \\
    General graphs   & \ok  & \ko  & \ok  & \ok  & \uk  & \ok  & \ok \\
    \bottomrule
    \end{tabular}
\end{table}

\begin{table}[tb]
    \newcommand{\ok}{\mathrel{\color{blue!80!black}\sim}}
    \newcommand{\ko}{\mathrel{\color{red!80!black}\nsim}}
    \newcommand{\va}{\radius}
    \newcommand{\vb}{\minvolume}
    \newcommand{\vc}{\avgvolume}
    \newcommand{\vd}{\maxvolume}
    \newcommand{\ve}{|\neigh_{\radius}|}
    \centering
    \caption{An overview of the hierarchy of different measures of mending: $\va$ is the mending radius as defined in~\cite{local_mending}, $\vb$ is the existential mending volume, $\vc$ is the expected mending volume, $\vd$ is the deterministic mending volume, and $\ve$ is the maximum number of nodes in a neighborhood of radius $\va$. In the table, $x \ok y$ indicates that $x$ and $y$ are asymptotically equivalent for all LCL problems, and $x \ko y$ indicates that there is at least one LCL problem for which this does not hold. Whether $\vb \ok \vc$ holds in trees is an open question.}
    \label{tab:model_comparison}
    \begin{tabular}{lccccccccc}
    \toprule
    Any graph family & $\va$ & $\le$ & $\vb$ & $\le$ & $\vc$ & $\le$ & $\vd$ & $\le$ & $\ve$ \\
    \midrule
    Paths and cycles & $\va$ & $\ok$ & $\vb$ & $\ok$ & $\vc$ & $\ok$ & $\vd$ & $\ok$ & $\ve$ \\
    Infinite trees   & $\va$ & $\ko$ & $\vb$ & $\ok$ & $\vc$ & $\ko$ & $\vd$ & $\ok$ & $\ve$ \\
    Trees            & $\va$ & $\ko$ & $\vb$ &       & $\vc$ & $\ko$ & $\vd$ & $\ok$ & $\ve$ \\
    General          & $\va$ & $\ko$ & $\vb$ & $\ko$ & $\vc$ & $\ko$ & $\vd$ & $\ok$ & $\ve$ \\
    \midrule
    Reference
    && \makebox[0pt]{Sect.~\ref{sec:volume_polynomial}}
    && \makebox[0pt]{Sect.~\ref{sec:separation_min_avg_general}}
    && \makebox[0pt]{Sect.~\ref{sec:separation_avg_max}}
    && \makebox[0pt]{Sect.~\ref{sec:landscape_maxvolume}}
    \\
    \bottomrule
    \end{tabular}
\end{table}

\section{Related work}

One of the first papers that make explicit use of the fact that some LCL problems have a logarithmic mending volume is by Panconesi and Srinivasan~\cite{panconesi95delta}. They compute a $\Delta$-coloring of a graph by recoloring an augmenting path of length up to $O(\log n)$ whenever there is a conflict. However, their main interest is solving the problem in a distributed message passing model and they therefore mainly focus on the mending radius instead of the mending volume of this problem.

The idea of refining the radius measure into a volume measure in the study of the landscape of LCL problems can be attributed to Rosenbaum and Suomela~\cite{far_vs_wide}, who show similarities and differences between the models. The volume complexity for LCL problems has been further refined by Grunau at al.~\cite{randomized_local_lll}. However, the focus of these papers is only on solvability (constructing a solution from nothing) rather than mendability (editing a partial solution to the closest complete solution) of a problem. They nevertheless highlight the fact that merely looking at the radius complexity does not capture all details of what information within that radius is actually necessary, and some problems that have the same radius complexity exhibit very different volume complexities. 

\subparagraph{Mending radius.}

Balliu et al.~\cite{local_mending} introduced the first formal graph-theoretic notion of mending radius. The authors show how to use mending as a tool for algorithm design and analyze the complexities of mending on paths, rooted trees and grids. 

In contrast to the definition of mending radius, our definition of mending volume captures more complexity classes of problems. A concrete example of the mending volume being more accurate than the mending radius is the group of three problems \(R_1\), \(R_2\), \(R_3\) defined in Problem~\ref{prob:R_i}.
\begin{problemdef}[b]
    \caption{\(R_i\)\problemcaptionkludge}
    \label{prob:R_i}
    \vspace{-\topsep}\begin{description}[noitemsep]
        \item[Input:] A balanced rooted ternary tree
        \item[Labels:] \textsf{red} and \textsf{white}
        \item[Task:] Color the vertices so that the root is \textsf{red}, and every \textsf{red} vertex has at least \(i\) \textsf{red} children.
    \end{description}\vspace{-\topsep}
\end{problemdef}
Assume that we start with a partial solution where the root is uncolored and all other nodes are colored \textsf{white}. Naturally, any mending algorithm must color the root \textsf{red}, and start updating the descendants of the root down to the leaves. The mending radius of all three problems is therefore $\Theta(\log_3 n)$. The mending volume, however, differs for all three problems $R_i$, and the corresponding complexities are discussed in Figure~\ref{fig:comparison_tree_coloring}.

\begin{figure}[tb]
    \includegraphics[width=\linewidth]{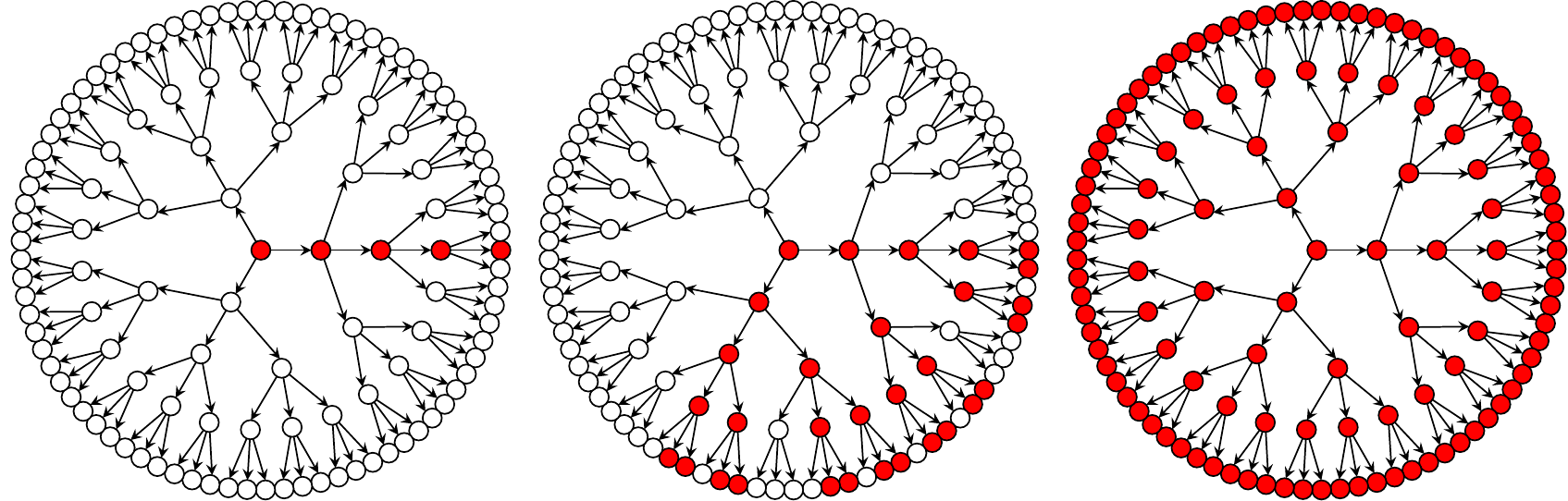}
    \caption{From left to right, solutions of \(R_1\), \(R_2\), \(R_3\) (as defined in Problem~\ref{prob:R_i}) with the least number of \textsf{red} labels are visualized. In the case of \(R_2\) (middle), each red vertex starting from the root in the center has two of its three children colored red, and this continues down to the leaves. The radius in this example is $4$ and its growth rate as the graph gets larger is \(\Theta(\log_3 n)\), the volume is \(2^{4+1} - 1\) which grows as \(\Theta(n^{\log 2 / \log 3})\). On each of these three solutions the set of vertices recolored \textsf{red} has the same radius \(\Theta(\log_3 n)\), yet the volume of the \textsf{red} zone is \(\Theta(\log_3 n)\) for \(R_1\), \(\Theta(n^{\log 2 / \log 3})\) for \(R_2\), and \(\Theta(n)\) for \(R_3\).} 
    \label{fig:comparison_tree_coloring}
\end{figure}

Also other papers made use of mending radius, mainly as an algorithm design tool. Chechik and Mukhtar~\cite{distributed_coloring} design an algorithm for $6$-coloring planar graphs using the observation that some small structures can be properly colored for any proper coloring of their surrounding vertices. Such structures can be removed from the graph temporarily while coloring the rest of the vertices. Similar observations have been made for computing a $\Delta$-coloring~\cite{panconesi95delta} and solving an edge-orientation with maximum out-degree $(1+\varepsilon)a$~\cite{10.1145/3465084.3467908}. 
Recently, it has been shown that mending algorithms with a constant radius can also be transformed into self-stabilizing algorithms in anonymous networks~\cite{self_stab_mending}.
On the other hand, there were also attempts to extract an explicit notion of mending, although using different definitions of partial solutions and complexity measures. This includes for example König and Wattenhofer~\cite{local_fixing}, and Kutten and Peleg~\cite{tight_fault_locality}.
König and Wattenhofer~\cite{local_fixing} consider only faults that are an addition or a deletion of a single vertex or edge at a time, and hence feature only at most a constant number of unlabeled vertices. Kutten and Peleg~\cite{tight_fault_locality} are interested in the time needed to compute a complete solution as a function of the initial number of failures.

\subparagraph{Local search.}

The idea of mending volume is closely related to local search in optimization problems (in the context of traditional centralized algorithms). Often one starts with a suboptimal solution and tries to converge to a better solution from there. Usually a problem is first solved by computing some possibly random initial variable assignment that satisfies the constraints, see e.g.~\cite{CHAMS1987260,GALINIER20062547,10.1007/978-3-540-74565-5_31}. Then, a local search algorithm is applied to find a better solution in the vicinity of the previous one.  

A classic application of local search in combinatorial optimization is the traveling salesman problem; local search is often applied to hard problems in order to achieve a good approximation of the optimal solution~\cite{10.2307/j.ctv346t9c, Angel2006}. On the negative side, Johnson et al.~\cite{JOHNSON198879} showed that an exponential number of iterations may be needed if the cost function can take exponential values. Ausiello and Protasi~\cite{AUSIELLO199573} later defined the class of guaranteed local optima (GLO) problems where the values of the cost function are bounded by a polynomial and showed that such problems can be solved in a polynomial number of iterations. Halld\'{o}rsson~\cite{10.5555/313651.313687} showed that local search can help to improve worst-case approximation guarantees by starting with a greedy solution and improving it locally using local search. He provides approximation results for various problems, such as the independent set, $k$-dimensional matching and $k$-set packing in nearly-linear sequential time.
Chandra and Halld\'{o}rsson~\cite{CHANDRA2001223} later showed an $2(k + 1)/3$ approximation algorithm for the weighted $k$-set packing problem, thus improving a previous result from Bafna et al.~\cite{BAFNA199641}, and Arkin and Hassin~\cite{doi:10.1287/moor.23.3.640}.

\section{Preliminaries}

Our definition of the mending volume is built along the lines of the definition of the mending radius in~\cite{local_mending}: we define the mending volume as a measure entirely independent of any distributed computing model and we place ourselves in the context of Locally Checkable Labeling problems (LCLs) first introduced in~\cite{Naor1995}. We use the same definition of partial solutions as~\cite{local_mending} in order to make our results comparable. A reader who is familiar with the notions of graph labeling problems---and LCLs in particular---as well as with the specific definition of partial solutions from~\cite{local_mending} may skip directly to the next section in which we introduce a formal definition of the mending volume.

\subsection{Locally checkable labelings}

LCLs are labeling problems on bounded-degree graphs. In these problems, an input graph with maximum degree \(\Delta = O(1)\) is given and the task is to produce an assignment of labels to vertices in a way that satisfies some predetermined local constraints. The specification of an LCL problem is done by means of a local verifier.
\begin{definition}[Local verifier]
    A \emph{verifier} \(\phi\) is a function that maps tuples \((G, \lambda, v)\) to \(\set{\textsf{happy}, \textsf{unhappy}}\), where \(v\) is a vertex and \(\lambda\) a labeling of \(G\). We say that the verifier \(\phi\) \textit{accepts} \(\lambda\) if \(\phi(G, \lambda, v) = \textsf{happy}\) for all \(v\), otherwise it \textit{rejects} \(\lambda\).

    In addition, \(\phi\) is local if, for some constant radius \(r\), whenever \((G_1,\lambda_1)\) and \((G_2,\lambda_2)\) coincide over the radius-\(r\) neighborhood of \(v_1\) and \(v_2\) then they have the same image according to \(\phi\). That is, \((G_1,\lambda_1)_{|\neigh_r(v_1)} \simeq (G_2, \lambda_2)_{|\neigh_r(v_2)}\) implies \(\phi(G_1, \lambda_1, v_1) = \phi(G_2, \lambda_2, v_2)\)
\end{definition}

An LCL problem is entirely characterized by a finite set of labels and a local verifier.

\begin{definition}[Locally Checkable Labeling]
    A \emph{Locally Checkable Labeling} problem \(\problem\) is represented by a finite set of labels \(\labels\), a class of input graphs \(\graph\), and a local verifier \(\phi\).
    An instance of \(\problem\) is a graph \(G\in\graph\).
    A solution is a labeling \(\lambda\) of \(G\) over \(\labels\) that is accepted by \(\phi\).
\end{definition}

\subsection{Partial solutions}

Mending takes as input an incomplete labeling and extends it into one that is a little closer to being complete. Since graph labelings were defined to be complete over all vertices, the most natural way to define partial solutions is to extend the set of labels with one fresh label \(\bot\) that is interpreted as ``unlabeled'', and adapt the local constraints to allow labelings that involve this new label. We will often refer to vertices that are labeled \(\bot\) simply as ``unlabeled vertices'' or ``holes''.

A desirable definition of partial solutions should satisfy the following three properties:
\begin{enumerate}
    \item A partial solution without any hole is a complete solution;
    \item the empty labeling (the constant function \(\lambda_\bot: \_ \mapsto \bot\)) is a partial solution;
    \item a sub-solution of a partial solution is also a partial solution. That is, if \(\lambda\) is a partial solution then any labeling \[\lambda_S: x\mapsto\left\{\begin{array}{ll} \lambda(x) & \text{if \(x\in S\)} \\ \bot & \text{otherwise} \\ \end{array}\right.\] is a partial solution.
\end{enumerate}

As stated in~\cite{local_mending}, the following is a simple way to satisfy all of these constraints: extend the verifier \(\phi'\) to be \textsf{happy} whenever an unlabeled vertex is visible in the radius-\(r\) neighborhood, otherwise fall back to the same rules as \(\phi\).

\begin{definition}[Partial solution]\label{def:partial_solution}
    For \(\problem = (\labels, \graph, \phi)\), where \(\phi\) has radius \(r\), define a relaxation \(\problem^* = (\labels^* \coloneqq \labels \sqcup \set{\bot}, \graph, \phi^*)\) of \(\problem\) to allow empty labels.

    For a labeling \(\lambda'\) over \(\labels^*\), define \(\phi^*(G, \lambda', v)\) as follows: if there exists a node \(u_\bot\) within distance \(r\) of \(v\) such that \(\lambda'(u_\bot) = \bot\), then \(\phi^*(G, \lambda', v) \coloneqq \textsf{happy}\); otherwise, let \(\lambda\) be any labeling over \(\labels\) that agrees with \(\lambda'\) on \(G_{|\neigh_v(r)}\) and set \(\phi^*(G, \lambda', v) \coloneqq \phi(G, \lambda, v)\).

    We define \(\dom_{\labels}(\lambda')\) to be the set of vertices that \(\lambda'\) labels with labels from \(\labels\). A labeling (resp. solution) of \(\problem^*\) is called a \emph{partial labeling} (resp. \emph{partial solution}) of \(\problem\).
\end{definition}

One can easily check that all the desirable properties stated above are satisfied by Definition~\ref{def:partial_solution}; this fact is also proven in~\cite{local_mending}. Note that this definition of partial solutions has a notion of locality that is consistent between labelings and partial labelings: the verifiers \(\phi\) and \(\phi^*\) have the same locality radius.

We can now define what it means to mend a partial solution: a mend of \(\lambda\) is a new partial solution with one specific vertex no longer labeled \(\bot\), and no additional \(\bot\) labels.

\begin{definition}[Mend]\label{def:mend}
    For a partial solution \(\lambda\) of \(\problem\) on an instance \(G\), we say that \(\lambda'\) is a \emph{mend} of \(\lambda\) at \(v\in G\) if the following hold:
    \begin{description}
        \item[Validity:]
            \(\lambda'\) is a partial solution.
        \item[Progress:]
            \(\dom_\labels(\lambda) \cup \set{v} \subseteq \dom_\labels(\lambda')\), that is, no \(\bot\) was added and \(v\) is no longer labeled \(\bot\).
    \end{description}
\end{definition}
The mending problem \(\mendprob(\problem)\) associated with an LCL \(\problem\) is the following task: given \(G \in \graph\), \(\lambda\) solution of \(\problem^*\) and \(v\) hole of \(\lambda\),
produce \(\lambda'\) a mend of \(\lambda\) at \(v\).

\section{Complexity landscape of existential mending volume}\label{sec:existential_landscape}

Having defined LCLs and partial solutions, we can now introduce mending volume. In this section, we consider an existential definition of the mending volume. This definition (see Section~\ref{sec:existential_volume_def}) is a purely graph-theoretic measure of the optimal solution for a worst-case instance of a mending problem. Later, in Section~\ref{sec:mending_rooted_trees}, we develop a technique for designing LCLs that have a specific existential mending volume on infinite rooted trees. In Section~\ref{sec:transfer_finite_trees}, we show that these problems can be transferred to finite and non-oriented trees while keeping the same mending volume complexity. Finally, in Sections~\ref{sec:volume_polynomial} and~\ref{sec:volume_polylog}, we apply these design techniques to obtain problems that have mending volume \(\Theta(n^\alpha), 0<\alpha<1\) or \(\Theta(\log^k n), k\in\mathbb{N}^*\), thereby providing examples of complexities that the mending volume exhibits that were not observed previously in the study of the mending radius.

\subsection{Existential mending volume: Definition}\label{sec:existential_volume_def}

For two labelings $\lambda$ and $\lambda'$, we define \(\diff(\lambda, \lambda') \coloneqq \set{v\!:\ \lambda(v) \ne \lambda'(v)}\) such that \(|\diff(\lambda,\lambda')|\) is the Hamming distance between two partial solutions. We define the existential mending volume of a problem \(\problem\) as the distance between the partial solution and the optimal mend for the worst-case instance \((G,\lambda,v)\) of \(\mendprob(\problem)\).
Here, \(G\) is an input graph from the family on which \(\problem\) is defined, \(\lambda\) is a partial solution, and \(v\) is a hole s.t. \(\lambda(v) = \bot\) at which \(\lambda\) must be mended.
\begin{definition}[Existential mending volume]\label{def:existential_volume}
    \[\minvolume(\problem) \coloneqq \max_{G,\lambda,v}\min\set{|\diff(\lambda,\lambda')|\!:\ \text{\(\lambda'\) mend of \(\lambda\) at \(v\)}}\]
\end{definition}

\subsection{Mending in infinite rooted trees}\label{sec:mending_rooted_trees}

In the following sections, we will establish a landscape of possible complexities that the existential mending volume can exhibit. For a summary, please refer to Table~\ref{tab:complexities_overview} that was introduced earlier. To this end, we describe examples of LCL problems that have logarithmic, polylogarithmic and polynomial existential mending volumes. The statement of these examples is made easier by the fact that all problems we show are of a specific type and we refer to them as \textit{propagation problems}. The complexity analysis of problems in this class is very straightforward for two reasons: (1) they admit a simple matrix description by an encoding shown in Section~\ref{sec:matrix_representation}, and (2) we only need to study their behavior in infinite trees thanks to results from Section~\ref{sec:transfer_finite_trees}. The advantage of infinite trees is that the complexity analysis is simplified by the absence of cycles, high-degree nodes, leaves, and other irregularities of the input graph. This restriction of only considering infinite regular rooted trees also has the complementary effect of illustrating that even simple problems already exhibit a rich variety of mending volume complexities. Since any propagation problem with mending volume \(T\) can be transformed into a problem on general trees or graphs with the same mending volume \(T\), our choice does not restrict the generality of our results.

\subsubsection{Propagation problems}

In this section, we define propagation problems on infinite rooted trees, with the goal to use them as a design tool for LCLs that exhibit specific mending volume complexities.

\begin{definition}[Infinite \(\Delta\)-regular rooted trees]
    We call trees that satisfy the following properties
    \begin{itemize}
        \item there are infinitely many vertices;
        \item exactly one vertex is distinguished as the root;
        \item every vertex admits a unique directed path to the root;
        \item every vertex has exactly \(\Delta\) incoming edges.
    \end{itemize}
    infinite \(\Delta\)-regular rooted trees, or simply infinite rooted trees when \(\Delta\) is clear from the context.
\end{definition}

Note that this class of graphs only consists of a single graph for a fixed \(\Delta\). On this class of input graphs, we define propagation problems as any LCL problem that is constructed according to the procedure explained in Definition~\ref{def:construct_propagation}.
\begin{definition}[Construction of a propagation problem]\label{def:construct_propagation}
    On the label set \(\labels\), distinguish two special labels---the initial label \(l_0\), and the wildcard label \(l_{-}\). Let \(\labels' \coloneqq\labels\setminus\set{l_-}\). Choose some \(\mu : \labels' \times \labels'\to \mathbb{N}\) and some \(\Delta \geq \max_{l\in\labels'} \sum_{l'\in\labels'} \mu(l,l')\).
    This defines an LCL on infinite \(\Delta\)-regular rooted trees, with locality 1, where the radius-1 neighborhood of \(v\) labeled by \(\lambda\) is accepted if all of the following constraints are satisfied:
    \begin{itemize}
        \item \(\lambda(v) = l_0\) if \(v\) is the root;
        \item if \(\lambda(v) = l \ne l_-\) then, for every \(l'\in\labels'\), there are at least \(\mu(l,l')\) children of \(v\) labeled \(l'\).
    \end{itemize}
\end{definition}

In other words, we only allow labeling constraints of the form ``any vertex labeled \(l\) must have at least \(\mu(l,l')\) children labeled \(l'\)'' or ``the root must be labeled \(l_0\)''. The requirement \(\Delta \geq \max_{l\in\labels'} \sum_{l'\in\labels'} \mu(l, l')\) is chosen such that all constraints are compatible with each other.
Note that there are no constraints involving the wildcard label \(l_-\) in this definition: it may appear as a child of any other label, and it may have any labels as its own children. In particular, the labeling where the root is unlabeled and all non-root vertices are labeled \(l_-\) is always a valid partial solution. We will show in Corollary~\ref{cor:matrix_worst_case} that this labeling is the worst-case instance for most propagation problems.

Since the input graphs on which these problems are defined are infinite, and since these problems often produce mends that have infinite volume, we study the volume not in terms of the total number of modified labels but in terms of the number of modified labels at distance at most \(d_{\max}\) from the hole.

\subsubsection{Matrix representation}\label{sec:matrix_representation}

Let \(M\) be a matrix of size \(|\labels'|\times|\labels'|\) defined as \(M[l,l'] = \mu(l,l')\). Observe that the coefficient \(M^d[l,l']\) of the $d$-th power of $M$ is the tightest lower bound on how many children labeled \(l'\) a vertex labeled \(l\) must have at distance \(d\) for a complete solution to be accepted. Indeed, by induction, a vertex labeled \(l\) must have at least \(M[l,l']\) children labeled \(l'\) at distance 1; it must then have at least \(\sum_{l''\in\labels'} M^d[l,l'']M[l'',l'] = M^{d+1}[l,l']\) children labeled \(l'\) at distance \(d+1\). We argue in Theorem~\ref{thm:matrix_bounds} that this provides bounds for \(\minvolume\).

We write \(\lVert L_l M^d \rVert \coloneqq \sum_{l'\in\labels'} M^d[l,l']\), where \(L_l\) is the line vector with a \(1\) only in position \(l\). We show in Theorem~\ref{thm:matrix_bounds} that this quantity expresses both upper and lower bounds on the mending volume up to distance \(d_{\max}\) of the propagation problem described by \(M\).
\begin{theorem}[Mending complexity of a propagation problem]\label{thm:matrix_bounds}
    The mending volume up to distance \(d_{\max}\) of a propagation problem represented by \(M\) is between \(\sum_{d=0}^{d_{\max}} \lVert L_{l_0} M^d \rVert\) and \(\max_{l\in\labels'}\sum_{d=0}^{d_{\max}} \lVert L_l M^d \rVert\)
\end{theorem}
\begin{proof}
    We start with the lower bound. Recall that in the input graph all vertices have degree exactly \(\Delta\). Consider an initial partial labeling \(\lambda\) in which the root is initially unlabeled, and all other vertices are labeled \(l_-\). A mend \(\lambda'\) of \(\lambda\) at the root will have to be a complete solution, and thus require the root to be labeled \(l_0\). By the previous observation, at distance \(d\), there must be at least \(M^p[l_0,l']\) vertices in \(\lambda'\) labeled \(l'\) that must have been modified during the mending. This way, \(\sum_{d=0}^{d_{\max}} \lVert L_{l_0} M^d \rVert\) is a lower bound for how many labels were modified at distance at most \(d_{\max}\).

    We can now show the upper bound. An important characteristic of the family of propagation problems is that the output of the verifier depends only on a portion of the labels of the children. Once sufficiently many children are labeled correctly, the remaining ones have no impact. This means that no initial configuration can force more than \(M^d[l,l']\) labels \(l'\) to be added at distance \(d\) from a vertex \(v\) labeled \(l\): in the worst case, it suffices to arbitrarily choose \(M[l,l']\) children at each level for each and color them accordingly while ignoring all the other children. Thus the worst-case instance has mending cost at distance \(d_{\max}\) no more than \(\max_{l\in\labels'} \sum_{d=0}^{d_{\max}} \lVert L_l M^d \rVert\)
\end{proof}

\begin{corollary}[Worst-case instance of a propagation problem]\label{cor:matrix_worst_case}
    If \(l_0\) is such that \(\lVert L_{l_0} M^d \rVert = \Omega(\max_{l\in\labels'} \lVert L_l M^d \rVert)\) then the initial instance where the root is unlabeled and all other vertices are labeled \(l_-\) is the worst-case instance.
\end{corollary}
The condition for Corollary~\ref{cor:matrix_worst_case} is satisfied at least for problems where all labels are reachable from \(l_0\) in the sense that for every \(l'\in\labels'\) there exists some \(d_{l'}\) for which \(M^{d_{l'}}[l_0,l'] \ne 0\).

\subsubsection{Landscape of the growth rate of matrix exponentiation}\label{sec:growth_rate_matrix}

In this section, we turn to a study of possible growth rates of the quantity \(\lVert L_l M^d \rVert \) introduced in Section~\ref{sec:matrix_representation}. This quantity is bounded by \(|\labels'|\times \max_{l'\in\labels'} M^d[l,l']\). In order to determine the mending volume of a propagation problem, it is sufficient to look at the growth rate as a function of \(d\) of \(\max_{l,l'\in\labels'} M^d[l,l']\)---the greatest coefficient of \(M^d\). We will denote it as \(\max M^d\).

In the following analysis, we make use of the interpretation of \(M\) as the adjacency matrix of a graph \(G_M\). \(G_M\) is a directed graph with one vertex for each element of \(\labels'\). For every pair \((v_l,v_{l'})\) there are exactly \(M[l,l']\) directed edges from \(v_l\) to \(v_{l'}\).
Further, there are \(M^d[l,l']\) walks of length exactly \(d\) from \(v_l\) to \(v_{l'}\) in \(G_M\). Let \(c(l)\) be the number of cycles in \(G\) that contain \(v_l\). We say that a vertex \(v_l\) is of type $0$ (\textit{resp.} $1$ or $2$) if \(c(l) = 0\) (\textit{resp.} \(c(l) = 1\) or \(c(l) \geq 2\)).
We will show that the type of the vertices fully determines the growth rate of \(|M^p|\): if some vertex is part of several cycles, then there are exponentially many paths of length \(d\) from that vertex to itself. Otherwise, if all vertices are part of at most one cycle, then there are only polynomially many paths of length \(d\) from one vertex to another.

\begin{lemma}
    Consider a vertex \(v_l\) of type 2. It holds that \(M^d[l,l] = \Omega((1+\beta)^d)\) for some \(\beta > 0\).
\end{lemma}
\begin{proof}
    Let \(C_1,C_2,\dotsc\) be the \(c(l)\) distinct cycles that contain \(v_l\).
    Let \(L_1,L_2,\dotsc\) denote their lengths respectively, and let \(L \coloneqq \lcm(L_1,L_2,\dotsc)\).
    There are at least \(c(l)\) walks of length \(L\) from \(v_l\) to itself, each following only one of the cycles \(C_j\) \ \(L/L_j\) number of times.
    Hence, for all \(k\), there are at least \(c(l)^k\) walks of length \(d := kL\) from \(v_l\) to itself and therefore \(M^d[l,l] \geq c(l)^{d/L}\) walks for infinitely many values of \(d\). Thus \(M^d[l,l] = \Omega((c(l)^{1/L})^d)\).
\end{proof}
\begin{corollary}
    Let vertex \(v_l\) be of type 2 and reachable from \(v_{l_0}\). Then \(M^d[l_0,l] = \Omega((1+\beta)^d)\)
    for some \(\beta > 0\).
\end{corollary}

\begin{lemma}
    If there is no vertex of type 2 reachable from \(v_{l_0}\) then \(|M^d|_{l_0} = O(d^k)\) for some constant \(k\).
\end{lemma}
\begin{proof}
    For each vertex \(v_l\), we denote \(C(l)\) to be the cluster of \(v_l\), defined as follows: \(C(l) \coloneqq \set{l}\) if \(c(l) = 0\); otherwise, \(C(l) \coloneqq \set{l' \mid v_l \to^* v_{l'} \to^* v_l}\) describes the vertices in the same cycle as \(v_l\). Since \(c(l) \leq 1\), the clusters form a partition of \(\labels'\). We use \(K\) to denote the number of clusters.

    Construct \(G'_M\) whose vertices are the clusters of \(G_M\), by contracting each cluster into a single vertex while keeping duplicate edges between different clusters, and removing edges inside a cluster. The resulting graph \(G'_M\) is acyclic.
    Any walk \(W\) of length exactly \(d\) in \(G_M\) from \(v_l\) to \(v_{l'}\) is uniquely defined by
    \begin{itemize}
        \item a walk \(W'\) in \(G'_M\) from \(C(l)\) to \(C(l')\), let \(C(l) = C_1 \to C_2 \to \cdots \to C_{|W|} = C(l')\) be this walk;
        \item the length \(d_i\) of the walk within \(C_i\), for each \(1\leq i\leq |W|\) (because no vertex is of type 2, there is only one such walk for a given length).
    \end{itemize}

    Note that \(d_1 + \cdots + d_{|W|} + (|W| - 1)\leq d\) and \(|W| \leq K \leq |\labels'|\). There are finitely many walks \(W'\) in \(G'_M\) and for each of them the number of possible tuples \((d_1,\dotsc,d_{|W|})\) is bounded by \(d^K\). Thus the number of walks \(W\) of length \(d\) is polynomially bounded by \(O(d^K)\).
\end{proof}

We observe further that if there is a walk in \(G'_M\) that goes through two or more cycles, then there are at least \(\Omega(d)\) paths of length \(d\). Whereas if \(G'_M\) contains only isolated cycles, then there are at most \(O(1)\) paths of length \(d\). Thus Theorem~\ref{thm:landscape_max_power} holds.
\begin{theorem}[Landscape of \(\max M^d\)]\label{thm:landscape_max_power}
    The growth rate of \(p \mapsto \max |M^d|\) is either eventually zero, or \(\Theta(1)\), or \(O(d^{\,p})\) for some value \(p \geq 1\), or \(\Omega((1+\beta)^d)\) for some \(\beta > 0\).
\end{theorem}
Combining these results with the known bounds from Theorem~\ref{thm:matrix_bounds} stating how to relate the mending volume to the growth of \(\max M^d\), we obtain Corollary~\ref{cor:landscape_infinite_trees}.
\begin{corollary}[Landscape of the mending volume on infinite rooted trees]\label{cor:landscape_infinite_trees}
    The mending volume up to distance \(d \coloneqq \log_\Delta n\) of a propagation problem is either
    \begin{itemize}
        \item \(O(1)\) if \(\max M^d\) is eventually zero;
        \item or \(\Theta(\log n)\) if \(\max M^d = \Theta(1)\);
        \item or \(O(\log^k n)\) for some \(k > 1\)  if \(\max M^d = O(d^{\,p})\);
        \item or \(\Omega(n^\alpha)\) for some \(0 < \alpha < 1\) if \(\max M^d = \Omega((1+\beta)^d)\).
    \end{itemize}
\end{corollary}

This concludes the survey of the landscape of propagation problems on infinite rooted trees. We found that there are complexity classes \(O(1)\), \(\poly(\log n)\) and \(\Omega(n^\alpha)\), with a gap between \(\omega(\log n)\) and \(o(\log^2 n)\). In Sections~\ref{sec:volume_polynomial} and~\ref{sec:volume_polylog} we will look more closely at the classes of growths \(\Omega(n^\alpha)\) and \(\Theta(\log^k n)\) to show that infinitely many values of \(\alpha\) and \(k\) can appear.

\subsection{Finite and non-regular trees}\label{sec:transfer_finite_trees}

Some of the presented results from Section~\ref{sec:growth_rate_matrix} are fortunately applicable to trees even if they are no longer infinite rooted and regular. Indeed there is a straightforward translation that transforms a propagation problem on infinite \(\Delta\)-regular rooted trees into one that has the same growth properties, but can be defined on finite rooted trees where some vertices are of degree lower than \(\Delta\). This construction is shown in Problem~\ref{prob:generalization_problem_trees}.
\begin{problemdef}[tb]
    \caption{Generalization of \(\problem\) to finite and non-\(\Delta\)-regular trees\problemcaptionkludge}
    \label{prob:generalization_problem_trees}
    \vspace{-\topsep}\begin{description}[noitemsep]
        \item[Input:] Any tree
        \item[Labels:] Same as \(\problem\)
        \item[Task:] Any vertex of degree exactly \(\Delta\) must satisfy the labeling constraints from \(\problem\)
    \end{description}\vspace{-1.1\topsep}
\end{problemdef}

We now argue that the fact that these trees are finite does not affect the conclusions made earlier about the possible complexities. The worst-case instance can namely still be constructed as a balanced finite \(\Delta\)-regular tree. We prove that (1) randomness does not decrease the performance and (2) the mending process is just as efficient in the case that the tree is unbalanced as it is in a balanced tree.

\begin{theorem}[Generalization to finite unbalanced trees does not change the mending volume complexity]
    \label{thm:unbalance}
    If \(\problem\) is a propagation problem, then its generalization \(\problem'\) to unbalanced trees has the same mending volume complexity.
\end{theorem}
\begin{proof}
    Assume that we wish to label a tree of size \(n+1\) rooted in \(v\). Each of the \(\Delta\) subtrees that are children of \(v\) have size \(n/\Delta+d_i\) for \(1 \leq i \leq \Delta\) where \(\sum_{i=1}^\Delta d_i = 0\), and we wish to assign a new labeling to each of them. The growth rate \(f^\text{bal}_j\) (\(1 \leq j \leq \delta\)) of the number of labels that would need to be modified for a balanced tree is uniquely determined by its assigned root label \(l_i\). A key observation is that as \(f^\text{bal}_j\) is defined by some \(\sum_{d=0}^{d_{\max}} \lVert L_l M^d \rVert \), it is either eventually zero or eventually convex.
    The total number of modified labels if the tree was balanced would be
    \[f^\text{bal}(n+1) = 1 + \sum_{i=1}^{\delta} f^\text{bal}_i(n/\delta).\]

    Assume inductively that the true number of modified labels in any non-balanced tree of size \(n'\) is less than \(f^\text{bal}_j(n')\), i.e. that a balanced tree is the worst-case input. Let \(c_{i,j}\) denote this number for the subtree \(i\) if it were assigned root label \(j\). The average performance of an algorithm that distributes the required labels randomly among all children with equal probability is then
    \begin{align*}
        c'
        &= \avg_{\sigma\in\mathfrak{S}_\delta} 1 + \sum_{i=1}^\delta c_{i,\sigma(i)} \\
        &\leq \avg_{\sigma\in\mathfrak{S}_\delta} 1 + \sum_{i=1}^\delta f^\text{bal}_{\sigma(i)}(n/\delta + d_i)
        &\text{induction hypothesis} \\
        &\leq 1 + \sum_{i=1}^\delta \avg_{\sigma\in\mathfrak{S}_\delta} f^\text{bal}_i(n/\delta + d_{\sigma(i)})
        &\text{reassign indices} \\
        &\leq 1 + \sum_{i=1}^\delta f^\text{bal}_i(n/\delta)
        &\text{by convexity of all \(f^\text{bal}_j\)} \\
        &\leq f^\text{bal}(n+1).
    \end{align*}

    The induction hypothesis holds for a balanced tree of size 1.
    Since the optimal mend has volume at most the expected volume of a mend picked at random, it follows that the existential mending volume on unbalanced trees is the same as the existential mending volume on balanced trees.
\end{proof}

\subsection{Application: \minvolumepolynomialtitle}\label{sec:volume_polynomial}

In the following two sections, we show examples of problems that exhibit polynomial and polylogarithmic complexities, with a particular focus on showing which values of \(0 < \alpha < 1\) and \(k \geq 1\) can appear for complexities \(\Theta(n^\alpha)\) and \(\Theta(\log^k n)\).

The prior analysis resulting in Corollary~\ref{cor:landscape_infinite_trees} suggests that, in order to construct a problem with volume \(n^{\Theta(1)}\), we should consider a propagation problem whose matrix \(M\) has exponential growth for \(d \mapsto \max M^d\). A good candidate is the problem described by \(M = \left(\begin{tabular}{c}2\end{tabular}\right)\) for \(\Delta = 3\). It describes the following problem:
\begin{problemdef}
    \caption{Polynomial propagation\problemcaptionkludge}
    \label{prob:polynomial_propagation}
    \vspace{-\topsep}\begin{description}[noitemsep]
        \item[Input:]  An infinite rooted \(\Delta\)-regular tree
        \item[Labels:]  \textsf{red} and \textsf{white}
        \item[Task:]  Color the vertices according to the following rules, by order of precedence:
        \begin{itemize}
            \item any labeling is valid for a vertex that does not have exactly 3 children;
            \item any labels are valid for the children of a vertex labeled \textsf{white};
            \item if a vertex is \textsf{red}, it needs at least two of its children to be \textsf{red};
            \item the root has to be \textsf{red}.
        \end{itemize}
    \end{description}\vspace{-0.8\topsep}
\end{problemdef}

Using the terminology of Definition~\ref{def:construct_propagation}, \textsf{red} is the initial label, and \textsf{white} is the wildcard label. From the initial labeling consisting of the root being unlabeled and all other vertices being \textsf{white}, a mend needs to recolor \(2^d\) vertices at layer \(d\) from \textsf{white} to \textsf{red}. For a balanced ternary tree, this will produce a total of \(2^{\log_3 n + 1} - 1\) recolored vertices, i.e. \(\Theta(n^{\ln 2 / \ln 3})\).

We further argue that by slightly adjusting the parameters, we can engineer any rational power of \(n\): the problem described by \(M = \left(\begin{tabular}{c}\(2^p\)\end{tabular}\right)\) for \(\Delta = 2^q\), where \(q > p > 1\), will exhibit complexity \(\Theta(n^{\ln 2^p / \ln 2^q}) = \Theta(n^{p/q})\).

\subsection{Application: \minvolumepolylogtitle}\label{sec:volume_polylog}

We now show how to construct a problem that has mending volume \(\log^k n\) for any chosen \(k \geq 1\). This time, Corollary~\ref{cor:landscape_infinite_trees} suggests to look for a matrix for which \(d \mapsto \max M^d\) has growth rate \(\Theta(p^k)\). This is satisfied by the matrix illustrated in Figure~\ref{fig:polylog_growth_matrix}: its size is \(k\), and it has entries \(1\) along and immediately above the diagonal, and entries \(0\) everywhere else. A solution to the problem described by this matrix has the following form: a path from a leaf to the root is labeled \(l_0\) including both ends. From each vertex labeled \(l_0\) there is a path labeled \(l_1\) from another leaf, and so on until each vertex labeled \(l_{k-2}\) being the endpoint of a path labeled \(l_{k-1}\) from a leaf.
\begin{figure}[tb]
    \centering\includegraphics[width=44mm]{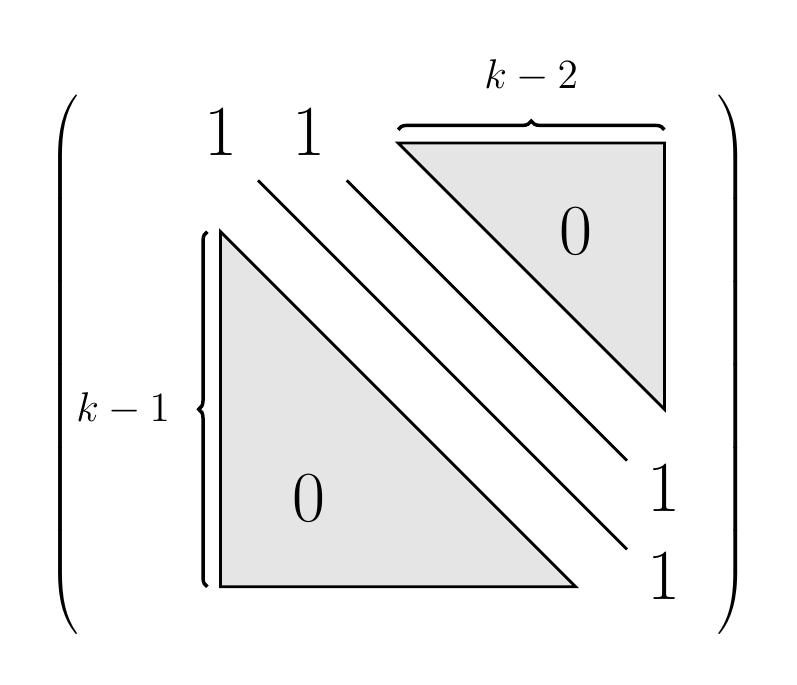}
    \vspace{-2mm}
    \caption{\(M_k\) of size \(k\times k\) exhibits growth \(\Theta(p^k)\) for the root label \(l_1\)}
    \label{fig:polylog_growth_matrix}
\end{figure}

\section{Algorithmic definitions of volume mending}\label{sec:algorithmic_mending_volume}

The mending volume introduced in Definition~\ref{def:existential_volume} and studied throughout Section~\ref{sec:existential_landscape} does not consider the complexity of finding a mend. For distributed systems, where no entity has a complete view of the input graph, a naive algorithm to compute the optimal mend may need to query an exponential number of vertices compared to how many of them will actually be relabeled in the end. For such applications, it may be more appropriate to consider alternative measures of mending where the cost is not just the number of labels that were modified but also the number of vertices that need to be queried before an algorithm with only local knowledge of the graph can compute a mend.

\subsection{Mending with local knowledge}

We will now focus on the process of computing the mend and make sure that each step of the computation can be completed with only local knowledge. A single-step definition like Definition~\ref{def:existential_volume} cannot express such restrictions. Therefore we will introduce a step-by-step definition of a process that computes a mend.

Such a process should take as input a graph \(G\), an initial partial solution \(\lambda\), and a hole \(v\) of \(\lambda\) that needs to be mended. We consider the following model describing the knowledge that the process has access to: initially it only knows \(W = \set{v}\) and the list of its direct neighbors. Whenever a vertex \(v'\) that is a direct neighbor of \(W\) is queried, the process can add \(v'\) to \(W\) and it acquires knowledge of all the neighbors of \(v'\). Thus at each step of the computation, the process has access to the connected set \(W\) of all vertices previously explored, and it can choose to explore any direct neighbor of \(W\) as its next vertex. Such an exploration model is similar to local computation algorithms discussed in~\cite{rubinfeld11} where we can learn the graph by probing it one node at a time.

The process can stop its exploration whenever the set of explored vertices \(W\) contains a mend in the following sense: there exists a mend of \(\lambda\) at \(v\) ---\(\lambda'\)--- in which only vertices from \(W\) are relabeled, i.e. such that \(\set{v'\ :\ \lambda'(v') \ne \lambda(v')} \subseteq W\). We choose to model this exploration process with a Markov chain: each state represents a possible value for \(W\). Having explored \(W\), the process may choose that its next vertex to explore should be \(v'\). In this case, the Markov chain will assign a nonzero probability to the transition \(W \to W\cup\set{v'}\).

\begin{definition}[Mender]\label{def:mender}
    For an input graph \(G = (V,E)\), an initial partial labeling \(\lambda\), and a hole \(v\) of \(\lambda\), let \(\mender_{G,\lambda,v}\) be a Markov chain over \(\parts(V)\) the subsets of \(V\).
    We call \(\mender\) a local \emph{mender} if the following properties hold for every \(G,\lambda,v\):
    \begin{description}
        \item[Progress:]
            \(\mender_{G,\lambda,v}(W,W') > 0\) implies \(W \subseteq W'\).
        \item[Termination:]
            \(\mender_{G,\lambda,v}(W,W) = 1\) iff there exists a mend of \(\lambda\) at \(v\) contained in \(W\).
        \item[Locality:]
            \(\mender_{G,\lambda,v}(W,\cdot)\) depends only on \(\neigh_1^G(W)\), i.e. \(W\) and its direct neighborhood.
    \end{description}
\end{definition}

Progress states that the set of explored vertices can only grow from one step to the next, and termination expresses that the computation will halt as soon as a mend is found. Locality ensures that at all steps of the computation the information accessible to the mender is exactly what it has already explored.

\begin{observation}[\(\mender_{G,\lambda,v}\) has a stationary distribution]
    Observe the following:
    \begin{enumerate}
        \item There must exist a mend of \(\lambda\) at \(v\in V\), so that \(V\) is absorbing. Since all other states of the Markov chain are subsets of \(V\), the conditions of \textsf{Termination} and \textsf{Progress} from Definition~\ref{def:mender} are compatible.
        \item All non-absorbing states eventually lead to an absorbing state in finitely many steps with nonzero probability since there are only finitely many subsets of \(V\) and hence finitely many states.
        \item There are no loops involving more than one state.
    \end{enumerate}
    From the above properties it follows that \(\mender_{G,\lambda,v}\) is an absorbing Markov chain, and therefore it has a stationary distribution that we denote \(\mender^*_{G,\lambda,v}\).
\end{observation}

\noindent We call all subsets of \(V\) that contain a mend \textit{final states}. Due to the existence of a stationary distribution, a final state \(W_F\) is reachable from the initial state \(W_0\) if and only if \(\mender^*_{G,\lambda,v}(W_0, W_F) > 0\). We write \(\final^{\mender}_{G,\lambda,v}(W_0)\) for the set of final states reachable from \(W_0\).

In Definition~\ref{def:complexity_generalized}, we formally define the three notions of mending volume. We therefore use the formalism of Markov chains developed in Definition~\ref{def:mender}.
These mending volume complexities are defined as measures of the sizes of the elements of \(\final^{\mender}_{G,\lambda,v}(W_0)\), where \(W_0 = \set{v}\) the hole to be mended. For all of these measures we consider the worst-case instance by taking a \(\max\) over all possible values of \(\lambda\) and \(v\).

\begin{definition}[Complexity]\label{def:complexity_generalized}
    Consider a mender \(\mender\) and a graph \(G\). We define the best-case volume, expected volume, worst-case volume, and radius as measures of \(\final^{\mender}_{G,\lambda,v}(\set{v})\):
    \begin{align*}
        \bm{\boldminvolume'(\mender, G)}
            &\coloneqq \max_{\lambda,v} \min\set{|W_F|\ :\ W_F\in\final^{\mender}_{G,\lambda,v}(\set{v})}
            &\text{best-case volume} \\
        \bm{\boldavgvolume(\mender, G)}
            &\coloneqq \max_{\lambda,v} \sum_{W_F} \mender^*_{G,\lambda,v}(\set{v},W_F) \cdot |W_F|
            &\text{expected volume} \\
        \bm{\boldmaxvolume(\mender, G)}
            &\coloneqq \max_{\lambda,v} \max\set{|W_F|\ :\ W_F\in\final^{\mender}_{G,\lambda,v}(\set{v})}
            &\text{worst-case volume} \\
        \bm{\boldradius(\mender, G)}
            &\coloneqq \max_{\lambda,v} \min_{W_F \in \final^{\mender}_{G,\lambda,v}(W)}\max_{w\in W_F}d(\set{v}, w)
            &\text{best-case radius}
    \end{align*}
\end{definition}

Intuitively, \(\minvolume'\) (best-case mending volume) is the size of the smallest reachable final state, \(\avgvolume\) (expected mending volume) is the expected size of an explored set given by the probability distribution of the mender, \(\maxvolume\) (deterministic mending volume) is the size of the largest reachable final state, and \(\radius\) (mending radius) is the smallest radius of a final state.

 We have also redefined the notion of existential mending volume \(\minvolume'\) using the formalism of Markov chains. As suggested by the notation, this measure coincides with \(\minvolume\) introduced in Definition~\ref{def:existential_volume}. We prove this fact in Lemma~\ref{lem:minvolume_equivalence} using a straightforward construction of a Markov chain that explores all possible mends so that the ``\(\min\)'' in the definition of \(\minvolume'\) results in the optimal mend.

 The worst-case volume is written \(\maxvolume\), i.e. \textit{Deterministic Mending Volume}, because an equivalent definition of \(\maxvolume\) is to remove all randomness from \(\mender\) and define \(\maxvolume\) as the deterministic size of the final state. This connection is explored in detail in Section~\ref{sec:landscape_maxvolume}. The definition of \(\radius\) also coincides with the mending radius introduced in~\cite{local_mending}, as is shown in Lemma~\ref{lem:radius_equivalence}.

Further, for a function \(T\) on integers and a family of input graphs \(\graph\), we say that \(\problem\) has \(\minvolume' = O(T)\), if there exists \(\mender\) that has performance \(\minvolume'(\mender, G_n) = O(T(n))\) on all \(G_n \in \graph_n\) instances of size \(n\). That is, there is a mender which performs at least as well as \(T\) asymptotically.
Similarly, we say that \(\problem\) has \(\minvolume' = \Omega(T)\) if for all \(\mender\) there exists \(G_n \in \graph_n\) such that \(\minvolume'(\mender, G_n) = \Omega(T(n))\), i.e., no mender can guarantee performance asymptotically better than \(T\) on all instances.
The above notation extends to all combinations of $o$, $O$, $\omega$, $\Omega$, $\Theta$ and $\minvolume'$, $\avgvolume$, $\maxvolume$, $\radius$.

Having concluded the definitions, we will next lay some foundations for a rough landscape of the existing mending complexities, extended with the new measures introduced in Definition~\ref{def:complexity_generalized}. The first step will be to establish the hierarchy between all the measures of complexity introduced previously, and to handle some simple cases such as the setting of paths and cycles.

\subsection{Hierarchy of complexities}\label{sec:hierarchy}

We first establish that \(\radius \leq 2r\times\minvolume\), where \(r\) is the radius of the local verifier that defines the problem \(\problem\). Consider an initial labeling \(\lambda\), and a mend of \(\lambda\) at \(v\) called \(\lambda'\). Denote \(D = \set{v\ :\ \lambda(v) \ne \lambda'(v)}\) the set of relabeled vertices. We argue that \(\neigh_r^G(D)\) is connected. Assume the opposite, and consider a maximal subset of \(D\) denoted \(D'\) that does not contain \(v\) where \(\neigh_r^G(D')\) is connected. Construct a new labeling $\lambda''$ as follows \[\lambda'': x \mapsto \left\{\begin{array}{ll}\lambda(x) & \text{if \(x\in D'\)} \\ \lambda'(x) & \text{otherwise.}\end{array}\right.\] Observe that \(\lambda'\) is a valid partial labeling: for every vertex, its neighborhood in \(\lambda''\) is identical to its neighborhood either in \(\lambda\) or in \(\lambda'\). In addition, \(\lambda''\) coincides with \(\lambda'\) over the neighborhood of \(v\). It is therefore a mend of \(\lambda\) at \(v\). Since \(\lambda''\) coincides with \(\lambda\) over \(D'\), it modifies fewer labels than \(\lambda'\). This shows that \(\lambda'\) is not minimal.

Therefore if \(\lambda'\) is a minimal mend then the radius-\(r\) neighborhood of the vertices it relabels is connected. By definition, any mend modifies a label at distance at least \(\radius\), therefore it must also modify at least one label every distance \(2r\) on a path from \(v\) to distance at least \(\radius\). This shows \(\radius \leq 2r\times\minvolume\).

Observe further that there is an easy strategy to ensure that \(\maxvolume \leq |\neigh_{\radius}|\): consider a naive mender \(\mender\) that always selects the next step to be the neighborhood of the current state with probability 1, i.e. \(\mender_{G,\lambda,v}(W, \neigh_1^G(W)) = 1\). By construction, it will halt after \(\radius\) steps on the final state \(\neigh_{\radius}^G(\set{v})\) since by definition the radius-\(\radius\) neighborhood always contains a mend.
Using the minimum, the expectation, and the maximum in Definition~\ref{def:complexity_generalized} leads to obvious inequalities \(\minvolume' \leq \avgvolume \leq \maxvolume\).

These observations provide the following hierarchy of mending complexities (we ignore constant multiplicative factors here): 
\begin{equation}\label{eq:complexity_hierarchy}
\radius \leq \minvolume \leq \avgvolume \leq \maxvolume \leq |\neigh_{\radius}|
\end{equation}

As shown in~\cite{local_mending}, the mending radius can only exhibit complexities on trees that are \(O(1)\), \(\Theta(\log n)\), or \(\Theta(n)\), and this is further restricted to \(O(1)\) or \(\Theta(n)\) on paths and cycles. Since our definition of \(\radius\) is equivalent, the same gaps hold. In the following settings, the hierarchy from Equation~\ref{eq:complexity_hierarchy} collapses:
\begin{itemize}
    \item if \(\radius\) is \(\Theta(n)\), so are all the other greater measures;
    \item if \(\radius\) is \(O(1)\), then \(|\neigh_{\radius}|\) is also \(O(1)\) and also all the smaller measures.
\end{itemize}

Hence the landscape is only diverse when the \(\radius\) is logarithmic and the other measures are between \(\Omega(\log n)\) and \(O(n)\). In Section~\ref{sec:separations}, we will exhibit separations between these measures on trees. They will necessarily rely on problems that have radius \(\Theta(\log n)\) since no separations exist otherwise.

\subsection{Separations}\label{sec:separations}

Equation~\ref{eq:complexity_hierarchy} shows a chain of inequalities between the different measures of mending that we have introduced. We can complement this equation by showing that all measures we have introduced are distinct from each other. This is done by proving that well-chosen LCL problems and classes of graphs feature \textit{separations} between these measures: on some instance, one measure of mending is arbitrarily smaller than another.

In Section~\ref{sec:separation_avg_max} we show that \(\avgvolume\) is distinct from \(\maxvolume\), and in Section~\ref{sec:separation_min_avg} we show that \(\minvolume\) is distinct from \(\avgvolume\).
These results are complemented by Section~\ref{sec:separations_appendix} in which we show that
\begin{itemize}
    \item \(\maxvolume\) is distinct from \(|\neigh_{\radius}|\) on balanced trees (Section~\ref{sec:separation_max_neigh});
    \item \(\radius\) is distinct from \(\minvolume\) using propagation problems (Section~\ref{sec:separation_rad_min});
    \item \(\minvolume\) is distinct from \(\avgvolume\) (Section~\ref{sec:separation_min_avg_general}); in this case, the separation does not rely on a promise unlike in Section~\ref{sec:separation_min_avg}, but it only works on general graphs instead of trees.
\end{itemize}

\subsubsection{\texorpdfstring{\(\avgvolume\)}{EMVol} is distinct from \texorpdfstring{\(\maxvolume\)}{DMVol} on trees and general graphs}\label{sec:separation_avg_max}
Most propagation problems described in Section~\ref{sec:existential_landscape}---for example the one with polynomial complexity from Section~\ref{sec:volume_polynomial}---feature a separation between \(\avgvolume\) and \(\maxvolume\) when executed on finite trees that are not necessarily balanced. Indeed when mending such a propagation problem on an unbalanced tree, the expected mending volume is the same as the existential mending volume as proved in Section~\ref{sec:transfer_finite_trees}. However there is no way to deterministically guarantee that the subtree explored is not a very unbalanced one, which could have size up to \(\Theta(n)\). This way, we have already shown several examples of problems with expected mending volume \(o(n)\) but deterministic mending volume \(\Theta(n)\).

\subsubsection{\minvolumetitle{} is distinct from \avgvolumetitle{} on trees with a global promise}\label{sec:separation_min_avg}

For this separation, we again resort to a promise. In fact we conjecture that on non-promise families of trees \(\minvolume\) is equivalent to \(\avgvolume\). The promise is as follows: we guarantee that the input tree is balanced and that there exists at least one vertex with a degree of exactly 2. The chosen problem is described in Problem~\ref{prob:deg_two_sink}.

\begin{problemdef}[ht]
    \caption{Degree-2 sink\problemcaptionkludge}
    \label{prob:deg_two_sink}
    \vspace{-\topsep}\begin{description}[noitemsep]
        \item[Input:] A balanced tree with at least one vertex of degree exactly 2 
        \item[Labels:] Orientations of the edges
        \item[Task:] Orient the edges so that vertices have out-degree 0 only if they have degree exactly 2.
    \end{description}\vspace{-1.1\topsep}
\end{problemdef}

A solution to this problem requires finding a specific vertex of a tree that has degree 2. Even with access to randomness, there is no algorithm that can guarantee to find such a vertex in fewer than \(\Theta(n)\) queries. On the other hand, if there is a guarantee that such a vertex exists, then it is at distance \(O(\log n)\) and can easily be found if one has access to a complete view of the graph. These properties make the \(\avgvolume\) linear, while the \(\minvolume\) as well as the \(\radius\) are only logarithmic.

As mentioned earlier, a similar separation result between \(\minvolume\) and \(\avgvolume\) is developed in Section~\ref{sec:separation_min_avg_general}. It uses general graphs instead of trees, but it also does not need to rely on a global promise unlike the separation on trees shown here.

\section*{Acknowledgments}
This work was supported in part by the Academy of Finland, Grant 333837.



\bibliography{literature}

\newpage
\appendix

\section{Separations}\label{sec:separations_appendix}

As announced in Section~\ref{sec:separations} we complement Equation~\ref{eq:complexity_hierarchy} by proving separations between the different measures of mending featured, as well as some equivalences between different definitions of the same measure.

\subsection{Equivalent definitions}
\label{sec:equivalent_defs}

The study of the landscape of the different measures of mending in Section~\ref{sec:hierarchy} assumes some equivalences between different definitions of the same measure. In particular we must prove that \(\minvolume\) introduced in Definition~\ref{def:existential_volume} is asymptotically equivalent to \(\minvolume'\) introduced in Definition~\ref{def:complexity_generalized}. We must also show that \(\radius\) from Definition~\ref{def:complexity_generalized} is equivalent to the notion of mending radius defined in~\cite{local_mending}.

The usage of ``\(\min\)'' in the definitions of \(\minvolume'\) and \(\radius\) suggests that as long as a process explores sufficiently many configurations, it will have \(\minvolume'\) and \(\radius\) close to the optimum. In particular, a mender that has all subsets as its reachable final states has the minimum possible size of a set that contains a mend as \(\minvolume'\), and the minimum radius as \(\radius\). Note that there exists such a mender, for example a mender that always explores one neighbor of its current explored set with equal probability.

\begin{lemma}[Equivalence of \(\minvolume\) and \(\minvolume'\)]\label{lem:minvolume_equivalence}
    \(\minvolume\) is asymptotically equivalent to \(\minvolume'\).
\end{lemma}
\begin{proof}
    We can easily show \(\minvolume = O(\minvolume')\): \(\minvolume\) is the minimum size of a mend, while \(\minvolume'\) is the minimum size of a connected set that contains a mend.

    Reciprocally by the same argument as in Section~\ref{sec:hierarchy}, the radius-\(r\) neighborhood of the set of relabeled vertices is connected, and it thus is of greater size than the set whose size is measured by \(\minvolume'\). Thus \(\minvolume' \leq \Delta^r\times\minvolume\).
\end{proof}

\begin{lemma}[\(\radius\) is the mending radius]\label{lem:radius_equivalence}
    \(\radius\) defined in Definition~\ref{def:complexity_generalized} is equivalent to the mending radius defined in~\cite{local_mending}.
\end{lemma}
\begin{proof}
    Assuming that all connected subsets of \(V\) are reachable states of the mender, \(\radius\) is then by definition the minimum radius of a mend.
\end{proof}

\subsection{\maxvolumetitle{} is distinct from \neighsizetitle{} on balanced trees }\label{sec:separation_max_neigh}

As we will show later in Theorem~\ref{thm:maxvolume_trivial}, the \(\maxvolume\) exhibits very few different complexities on trees, and it is always asymptotically of the same size as \(\neigh_{\radius}\). As such a separation between \(\maxvolume\) and \(|\neigh_{\radius}|\) needs to rely on a promise on the structure of the graph. A sufficient promise is that the input tree is a balanced binary tree. We will consider the problem of sinkless orientation, defined in Problem~\ref{prob:SO_balanced_trees}.
\begin{problemdef}[b]
    \caption{Sinkless orientation on balanced trees\problemcaptionkludge}
    \label{prob:SO_balanced_trees}
    \vspace{-\topsep}\begin{description}[noitemsep]
        \item[Input:] A balanced binary tree
        \item[Labels:] Orientations of the edges
        \item[Task:] Orient the edges so that all vertices with degree at least 2 have out-degree at least 1.
    \end{description}\vspace{-1.1\topsep}
\end{problemdef}
Note that the verifier has radius 1, where valid configurations are those where either the vertex has at most one neighbor, or it has out-degree at least 1. A mending procedure works as follows: starting from the initial hole, explore exactly one path to any leaf. This path can be mended simply by ensuring that all of its edges are directed towards the leaf. Since any leaf is acceptable, the explored path can be computed locally in a deterministic manner, and it has size \(O(\log n)\) as all leaves are at an at most logarithmic distance from the hole.

A worst-case instance is one where all edges are oriented towards the hole and labeling the hole would produce a sink. This case can require that at least the path from the hole to the nearest leaf is relabeled, which produces a mending radius of at least \(\Omega(\log n)\) if all leaves are at an at least logarithmic distance from the hole.

Thus sinkless orientation on balanced trees has \(\maxvolume(\problem) = \Theta(\log n)\) whereas \(|\neigh_{\radius}| = \Theta(n)\)

\subsection{\radiustitle{} is distinct from \minvolumetitle{} on trees and general graphs}\label{sec:separation_rad_min}

In Sections~\ref{sec:volume_polylog} and~\ref{sec:volume_polynomial} we have constructed problems that exhibit \(\minvolume\) either \(\Theta(\log^k n)\) or \(\Theta(n^\alpha)\) while having \(\radius = \Theta(\log n)\). These problems illustrate a separation between \(\radius\) and \(\minvolume\).

\subsection{Separating \avgvolumetitle{} from \minvolumetitle{} without a promise}\label{sec:separation_min_avg_general}

We observed that \(\avgvolume\) is equal to \(\minvolume\) in the case of propagation problems on both finite and infinite rooted trees. We will next show that \(\avgvolume\) and \(\minvolume\) can be separated in general graphs. We will therefore construct a family of directed graphs and a simple problem that together lead to existential mending volume complexity \(\minvolume = O(\log n)\) and expected mending volume complexity \(\avgvolume = \omega(\log n)\). We then show that this family of graphs can be extended to all undirected graphs while preserving the upper bound on \(\minvolume\). This construction will also naturally preserve the lower bound on \(\avgvolume\).

\subsubsection{Problem description}\label{sec:separation_problem}

We first define the family of graphs \(\tree = \bigcup_{h\in\NN} \tree_h\) built by adding horizontal layers to a tree skeleton in the following way.
\begin{definition}[\(\tree\)]\label{def:layered_trees}
    Start with the balanced binary rooted tree \(t_h\) that has height \(h+1\) for any \(h\in\mathbb{N}\). We will assign coordinates to each of the vertices, consisting of the distance from the root and the horizontal index. The root gets coordinate \((0,0)\) and the children of \((i,j)\) are assigned coordinates \set{(i+1, 2j), (i+1, 2j+1)}.

    We add an undirected edge between the sibling vertices \((i,j) - (i,j+1)\) for every \(0 \leq i < h\) and \(0 \leq j < 2^i - 1\). Additionally, consider a vertex \((h,j_0)\) on the bottom layer for some \(0 \leq j_0 < 2^i\) and add a directed edge from \((h,j) \to (h,j+1)\) for every \(0 \leq j < j_0\) and from \((h,j) \leftarrow (h,j+1)\) for every \(j_0 \leq j < 2^i-1\).

    This way, the pair \((h,j_0)\) uniquely defines one graph \(t_{h,j_0}\) with \(h\) layers and all edges on the bottom layer pointing towards the vertex at position \(j_0\). We call the vertex \((h,j_0)\) the \textit{sink} of \(t_{h,j_0}\).

    The infinite family of graphs $\tree$ is then defined as \(\tree \coloneqq \bigcup_{h\in\NN} \tree_h = \bigcup_{h\in\NN} \set{t_{h,j_0} \mid 0 \leq j_0 < 2^h}\).
\end{definition}
One such graph is illustrated in Figure~\ref{fig:example_layered_tree}, with a height of 5 and a sink at position 4.

\begin{figure}[tb]
    \centering\includegraphics[width=0.6\textwidth]{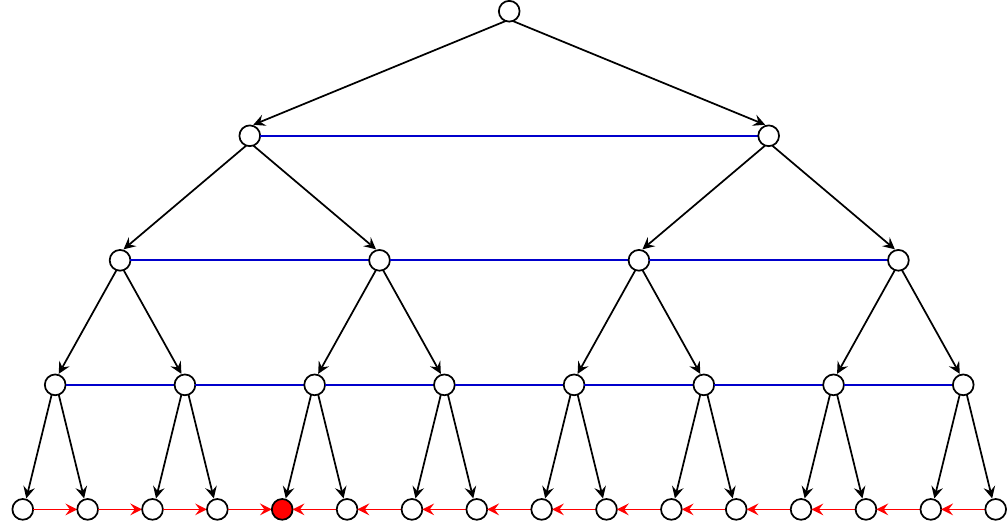}
    \caption{A visualization of a \(t_{5,4}\). Elements of \(\tree\) are entirely characterized by their height (here 5) and the position of the sink (marked red at position 4). The \(t_{5,4}\) features a skeleton in the form of a balanced binary oriented tree, to which horizontal unoriented links (in blue) are added. On the bottom layer, there are oriented edges (in red) towards the sink.}
    \label{fig:example_layered_tree}
\end{figure}

For this family of graphs \(\tree\), we define Problem~\ref{prob:path_to_sink}.
\begin{problemdef}[b]
    \caption{Path to the sink\problemcaptionkludge}
    \label{prob:path_to_sink}
    \vspace{-\topsep}\begin{description}[noitemsep]
        \item[Input:]  Some \(t_{h,j_0} \in \tree\)
        \item[Labels:] \textsf{red} and \textsf{black}
        \item[Task:] Color the vertices so that
        \begin{itemize}
            \item the root is \textsf{red}
            \item any \textsf{red} vertex that is not the sink has at least one \textsf{red} child.
        \end{itemize}
    \end{description}\vspace{-1.1\topsep}
\end{problemdef}
Observe that for any graph of \(\tree\), there is exactly one coloring that is complete and satisfies the labeling constraints from Problem~\ref{prob:path_to_sink}. This case is visualized in Figure~\ref{fig:t_54_unique_col} for the input graph \(t_{5,4}\). In this coloring, all vertices are colored white except for the unique path from the root to the sink. Note that there are many valid incomplete labelings as visualized in Figures~\ref{fig:t_54_partial_col_1}--\ref{fig:t_54_partial_col_3}. If some leaves are unlabeled, then the vertices above them can be colored \textsf{red}, and if the root is unlabeled then the entire graph can be colored \textsf{black}.

\begin{figure}[tb]
    \begin{subfigure}[b]{0.45\textwidth}
        \centering
        \includegraphics[width=\textwidth]{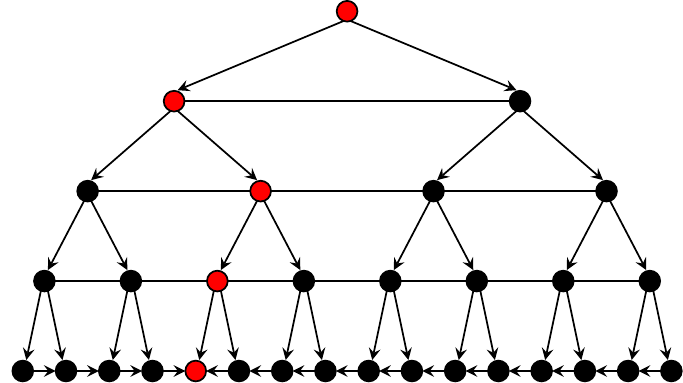}
        \caption{The unique valid complete labeling of \(t_{5,4}\)\\}
        \label{fig:t_54_unique_col}
    \end{subfigure}
    \hfill
    \begin{subfigure}[b]{0.45\textwidth}
        \centering
        \includegraphics[width=\textwidth]{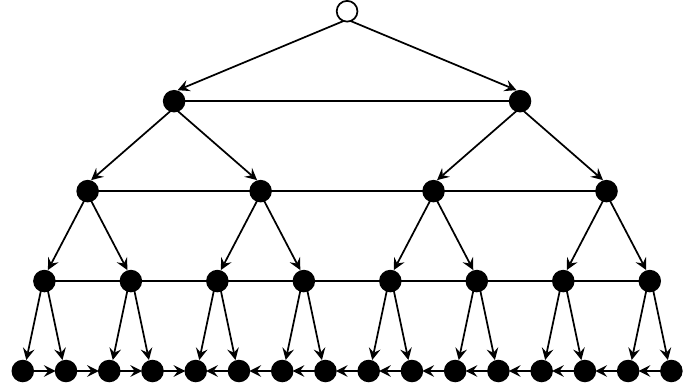}
        \caption{Partial coloring where all nodes are black besides the uncolored root.}    
        \label{fig:t_54_partial_col_1}
    \end{subfigure}
    \begin{subfigure}[b]{0.45\textwidth}
        \centering
        \includegraphics[width=\textwidth]{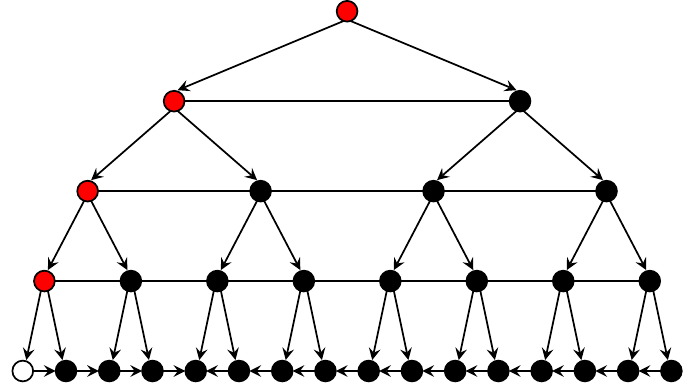}
        \caption{Partial coloring with one red path down to a child.}
        \label{fig:t_54_partial_col_2}
    \end{subfigure}
    \hfill
    \begin{subfigure}[b]{0.45\textwidth}
        \centering
        \includegraphics[width=\textwidth]{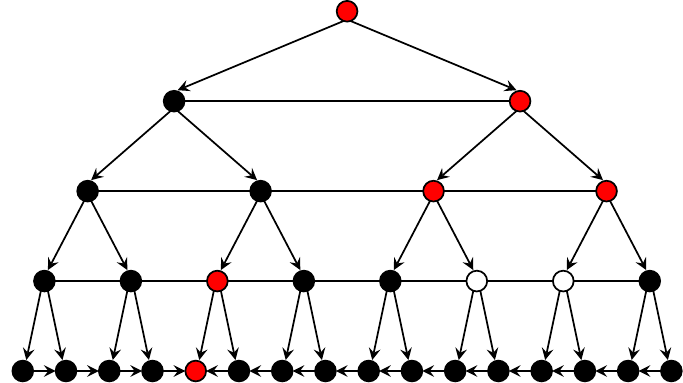}
        \caption{Partial coloring with several red subtrees.\\}
        \label{fig:t_54_partial_col_3}
    \end{subfigure}
    \caption{Example of four possible colorings of \(t_{5,4}\).}
\end{figure}

We will show in Sections~\ref{sec:minvolume_of_path_to_sink} and~\ref{sec:avgvolume_of_path_to_sink} that this problem has logarithmic \(\minvolume\), yet its \(\avgvolume\) is \(\omega(\log n)\). We conjecture the complexity of \(\avgvolume\) to even be \(\Omega(\log^2 n)\). This shows that \(\minvolume\) and \(\avgvolume\) are distinct from each other on oriented graphs, which is a separation that is not observed on trees and proven to not exist on rooted trees. This result extends to unoriented graphs thanks to an encoding of oriented graphs described in Section~\ref{sec:unorient}.

\subsubsection{Study of \minvolumetitle}\label{sec:minvolume_of_path_to_sink}

We now show that Problem~\ref{prob:path_to_sink} has logarithmic \(\minvolume\). A mending algorithm that uses a complete view of the graph to compute a mend with volume \(O(\log n)\) is described in Algorithm~\ref{alg:mending_path_to_sink}.
\begin{algorithm}[tb]
    \begin{algorithmic}[1]
        \Input
            \Desc{\(t\)}{input graph from \(\tree\)}
            \Desc{\(\lambda\)}{initial labeling}
            \Desc{\(v\)}{hole of \(\lambda\)}
        \EndInput
        \State \(\textsf{current} \gets v\)
        \State \textbf{\textsf{Phase 1}}
        \Loop
            \State \(\lambda(\textsf{current}) \gets \textsf{black}\)
            \State \(\textsf{parent} \gets \text{parent of \(\textsf{current}\)}\)
            \State \(\textsf{other} \gets \text{other child of \(\textsf{parent}\)}\)
            \If{\(\lambda(\textsf{parent}) = \textsf{black}\)
            \Or \(\lambda(\textsf{parent}) = \bot\)
            \Or \(\lambda(\textsf{parent}) = \lambda(\textsf{other}) = \textsf{red}\)
            }
                \State \Return
            \ElsIf{the sink is a descendant of \textsf{parent}}
                \State \(\textsf{current} \gets \textsf{parent}\)
                \State go to \textsf{Phase 2}
            \Else
                \State \(\textsf{current} \gets \textsf{parent}\)
            \EndIf
        \EndLoop
        \State \textbf{\textsf{Phase 2}}
        \State Recolor \textsf{red} a path from \textsf{current} to the sink
    \end{algorithmic}
    \caption{Mending procedure of Problem~\ref{prob:path_to_sink}}
    \label{alg:mending_path_to_sink}
\end{algorithm}

\begin{lemma}[Algorithm~\ref{alg:mending_path_to_sink} is correct]
    When Algorithm~\ref{alg:mending_path_to_sink} terminates, \(\lambda\) is a valid mend of the initial input labeling.
\end{lemma}
\begin{proof}
    Observe that Algorithm~\ref{alg:mending_path_to_sink} immediately assigns a label to the hole that it mends, and it never writes a \(\bot\) label. Therefore Algorithm~\ref{alg:mending_path_to_sink} is correct if and only if it outputs a valid labeling. An invariant of Algorithm~\ref{alg:mending_path_to_sink} is that at the beginning of any iteration of the main loop, all vertices other than \textsf{current} have a valid labeling.

    This invariant is preserved: recoloring the \textsf{current} vertex \textsf{black} makes \textsf{current} correctly labeled, and can at most introduce an incompatibility in the labeling of its \textsf{parent}. In such a case, we immediately perform \(\textsf{current} \gets \textsf{parent}\) which restores the invariant.
    The situations in which relabeling \textsf{current} does not introduce an inconsistency in the neighborhood of \textsf{parent} are the following: \textsf{parent} is already \textsf{black}, or \textsf{parent} is \textsf{red} and it has another \textsf{red} child. In these situations, the execution terminates with a valid labeling.

    Phase 2 is only ever reached in the following situation: \textsf{current} is labeled \textsf{red} with two \textsf{black} children, the sink is one of its descendants, and all vertices other than \textsf{current} are correctly labeled. In this situation, relabeling a path from \textsf{current} to the sink to \textsf{red} restores all labeling constraints and results in a valid mend.
\end{proof}

Both phases terminate in no more than \(2h = 2 \log_2 n\) relabeling steps and thus the problem has \(\minvolume = O(\log n)\).

\subsubsection{Simulating the problem in general oriented graphs}\label{sec:path_to_sink_general_graphs}

We now transfer the problem to work on general graphs without any promise on the structure of the input. As a first step, we will keep the orientation and assume that the directed edges as well as the parent-child relationship are part of the graph. We first show that the class \(\tree\) from Definition~\ref{def:layered_trees} can be defined by local constraints on directed graphs. These constraints are listed in Problem~\ref{prob:structural_constraints}: the graphs that satisfy all of these constraints are exactly the elements of \(\tree\).

\begin{problemdef}[tb]
    \caption{Structural constraints for \(\tree\)\problemcaptionkludge}
    \label{prob:structural_constraints}
    \vspace{0.2\topsep}
    In a graph with some oriented edges, given a vertex \(v\), we define the following constraints
    \vspace{-0.5\topsep}
    \begin{description}[noitemsep]
        \item[1a.] \(v\) has either zero or one parent \(u\) and there is an edge \(u \to v\)
        \item[1b.] \(v\) has either zero or two children \(u_1,u_2\) and there are edges \(v \to u_i\)
        \item[2a.] \(v\) has no siblings if and only if it has no parent
        \item[2a'.] \(v\) has exactly one sibling only if its parent has exactly one sibling
        \item[2a''.] \(v\) has no more than two siblings
        \item[2b.] if \(v\) has children then they are siblings
        \item[2c.] if \(v\) has children \(v_1\) and \(v_2\) and \(u\) is a sibling of \(v\), then \(u\) has children \(u_1\) and \(u_2\) and exactly one of \(u_1\) or \(u_2\) is linked to exactly one of \(v_1\) or \(v_2\)
        \item[2c'] if \(v\) is a sibling of \(u\) then either they have the same parent, or their parents are siblings
        \item[3a.] the edges of \(v\) to its siblings are oriented if and only if \(v\) has no children
        \item[3b.] \(v\) does not have edges oriented outwards to both of its siblings
    \end{description}\vspace{-0.8\topsep}
\end{problemdef}

We call a vertex \emph{broken} if it does not satisfy all the constraints from Problem~\ref{prob:structural_constraints}. A (sub)graph that has no broken vertices is called \emph{well-formed}. Note that all graphs of \(\tree\) are well-formed. Conversely, a non-empty well-formed graph is in \(\tree\), as shown by the construction illustrated in Figure~\ref{fig:construct_tree}. Constraints \textsf{1a} and \textsf{1b} from Problem~\ref{prob:structural_constraints} guarantee that the graph has a binary tree structure, though it is not yet guaranteed to be balanced. Constraints \textsf{2b} and \textsf{2c} impose a unique way of adding edges between sibling vertices, and guarantee that the tree is balanced. The resulting structure can no longer be modified: \textsf{2a} forbids any additional sibling of the root (shown in red in Figure~\ref{fig:construct_tree_4}); \textsf{2a'} similarly prevents internal vertices from having additional siblings (shown in yellow); finally there are no available vertices that could be siblings of the vertices on the border (shown in blue) because of constraints \textsf{2c} and \textsf{2c'}. Constraints \textsf{3a} and \textsf{3b} then impose an orientation of the edges between sibling vertices of the bottom layer once a sink is chosen.

\begin{figure}[tb]
    \begin{subfigure}[b]{0.45\textwidth}
        \centering
        \includegraphics[width=\textwidth]{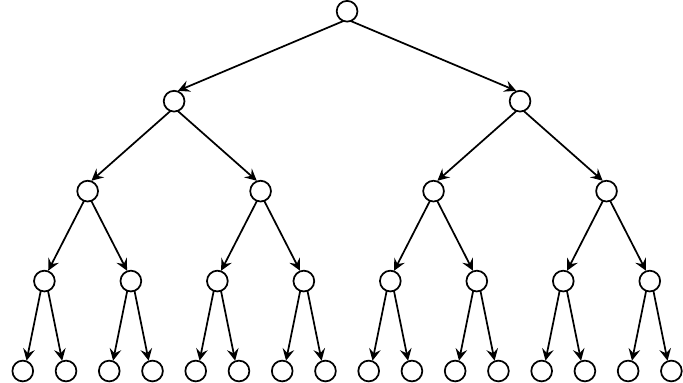}
        \caption{Step 1 -- apply constraints 1a and 1b}
        \label{fig:construct_tree_1}
    \end{subfigure}
    \hfill
    \begin{subfigure}[b]{0.45\textwidth}
        \centering
        \includegraphics[width=\textwidth]{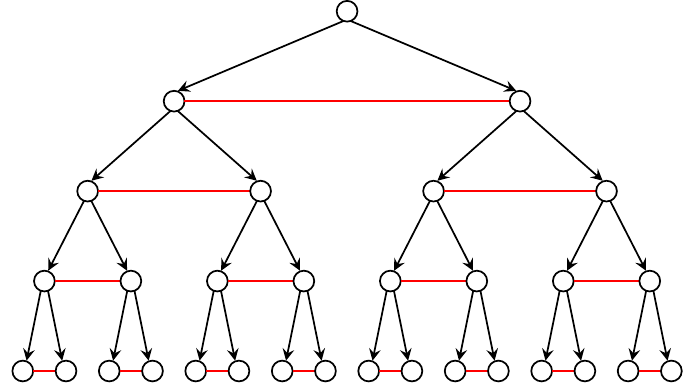}
        \caption{Step 2 -- apply constraint 2b}
        \label{fig:construct_tree_2}
    \end{subfigure}
    \begin{subfigure}[b]{0.45\textwidth}
        \centering
        \includegraphics[width=\textwidth]{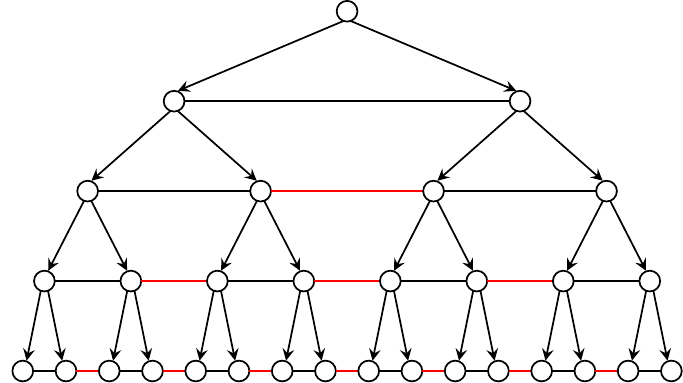}
        \caption{Step 3 -- apply constraint 2c}
        \label{fig:construct_tree_3}
    \end{subfigure}
    \hfill
    \begin{subfigure}[b]{0.45\textwidth}
        \centering
        \includegraphics[width=\textwidth]{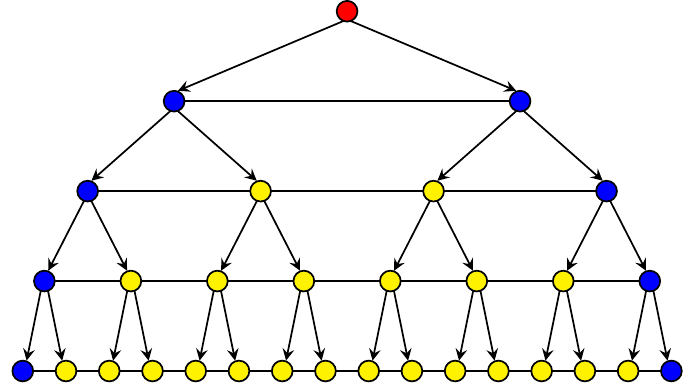}
        \caption{Step 4 -- no more links can be added}
        \label{fig:construct_tree_4}
    \end{subfigure}
    \caption{The construction of a well-formed graph results in a \(t_{h,j_0}\): at each step of the construction, there is only a single way to add edges in a way that satisfies the imposed local constraints.}
        \label{fig:construct_tree}
\end{figure}

It should be noted that a well-formed subgraph is not necessarily in \(\tree\). Any subgraph of a well-formed graph is well-formed, and there are many proper subgraphs from graphs of \(\tree\) that are not themselves in \(\tree\).  However, for any well-formed subgraph \(G' \subseteq G\), there exists a \(t_{h',j_0'} \in \tree\) of which \(G'\) is a subgraph.

In general graphs, we extend the problem so that graphs outside of \(\tree\) can also be labeled. In addition to the previous labeling rules from Problem~\ref{prob:path_to_sink}, any labeling is valid for a vertex that is broken (see Problem~\ref{prob:path_to_sink_oriented}).
\begin{problemdef}[tb]
    \caption{Path to the sink -- oriented graphs\problemcaptionkludge}
    \label{prob:path_to_sink_oriented}
    \vspace{-\topsep}\begin{description}[noitemsep]
        \item[Input:] An oriented graph with a parent-child relationship
        \item[Labels:] \textsf{red} and \textsf{black} (same as in Problem~\ref{prob:path_to_sink})
        \item[Task:] Vertices that satisfy all structural constraints from Problem~\ref{prob:structural_constraints} must also satisfy the labeling constraints from Problem~\ref{prob:path_to_sink}
    \end{description}\vspace{-1.1\topsep}
\end{problemdef}

\begin{lemma}[Preservation of \(\minvolume\) on oriented graphs]
    Problem~\ref{prob:path_to_sink_oriented} has the same existential mending volume as Problem~\ref{prob:path_to_sink}, namely \(\Theta(\log n)\).
\end{lemma}

\begin{proof}
    The key property to the preservation of \(\minvolume = O(\log n)\) is the following: for any vertex, there exists a vertex at distance \(O(\log n)\) that is either a sink or broken.

    A modified version of Algorithm~\ref{alg:mending_path_to_sink} where ``the sink is a descendant of \textsf{parent}'' is replaced with ``the sink or a broken vertex is a descendant of \textsf{parent}'' terminates when a sink or a broken vertex is found, and it leaves a valid labeling.

    We now argue that it terminates in time \(O(\log n)\). If no sink or broken vertex is found among the descendants, then it implies the existence of a balanced binary tree of non-broken vertices rooted at the current vertex. Each time we move to the parent of the current vertex, the size of the corresponding tree doubles. Thus only \(O(\log n)\) many iterations can be done before the size of such a tree would exceed the total number of available vertices, and thus after \(O(\log n)\) iterations there has to be either a sink or a broken vertex among the descendants of the current vertex. Once there is such a descendant, it must be at distance \(O(\log n)\).
\end{proof}

Therefore Problem~\ref{prob:path_to_sink_oriented} has \(\minvolume = O(\log n)\) in general oriented graphs. Section~\ref{sec:unorient} will further strengthen this result by showing that the orientation can be encoded in unoriented graphs. A similar technique can be applied to eliminate the need for the parent-child relation in the input graph.

\subsubsection{Study of \avgvolumetitle}\label{sec:avgvolume_of_path_to_sink}

We consider an arbitrary algorithm that relies on randomness and its local view of the graph in order to find a mend. Throughout this study, we assume that the graph is among \(\tree_h\). Assuming a graph of this structure, the performance of the algorithm is at least as good as the performance if such an assumption is not available. Note that any mend of the worst-case initial configuration, i.e., a configuration which consists of an unlabeled root and where all other vertices are \textsf{black}, must contain a path from the root to the sink that has to be recolored. Hence any algorithm that computes a mend must be able to find the sink. We show that the simple fact of finding the sink is a problem that cannot be solved in \(O(\log n)\).

Consider an execution of the algorithm as the sequence of all explored vertices, and call the subsequence corresponding to the interval between the visit of two vertices at the bottom layer a \emph{query}. Since the vertices of the bottom layer are the only vertices holding information on the position of the sink, the algorithm can be assumed to have a binary-search-like structure as presented in Algorithm~\ref{alg:finding_sink}.

\begin{algorithm}[tb]
    \begin{algorithmic}[1]
        \Input
            \Desc{\(t\)}{input graph from \(\tree\)}
            \Desc{\(h\)}{height of \(t\)}
        \EndInput
        \State \(\textsf{low} \gets 1\); \(\textsf{high} \gets 2^h\)
        \While{\(\textsf{low} < \textsf{high}\)}
            \State choose \(\textsf{low} \leq l \leq \textsf{high}\)
            \State query \(l\) by exploring a path from any already explored vertex
            \If{the right edge of \(l\) is oriented away from \(l\)}
                \Comment the sink is right of \(l\)
                \State \(\textsf{low} \gets l\)
            \ElsIf{the left edge of \(l\) is oriented away from \(l\)}
                \Comment the sink is left of \(l\)
                \State \(\textsf{high} \gets l\)
            \Else
                \Comment the sink is \(l\)
                \State \(\textsf{low},\textsf{high} \gets l\)
            \EndIf
        \EndWhile
    \end{algorithmic}
    \caption{Generic structure of any algorithm for finding the sink in a graph of \(\tree\)}
    \label{alg:finding_sink}
\end{algorithm}

Algorithm~\ref{alg:finding_sink} describes a generic structure that fits any optimal algorithm able to find the sink in a graph from the family \(\tree\). Indeed since the only information about the position of the sink is recorded in the orientation of the edges on the bottom layer, the algorithm will necessarily explore a certain number of vertices from the bottom layer before it finds the sink. We can further assume that if the algorithm is optimal then it performs no useless query, i.e. if the algorithm has already determined that the sink is right of a certain vertex \(l\) then it will never again query vertices left of \(l\). 

We will separate the queries into two categories.
\begin{definition}[Short and long queries]
    In Algorithm~\ref{alg:finding_sink}, if vertex \(l\) is queried while \(\min(\textsf{high} - l, l - \textsf{low}) < h\), we call the query \emph{short}. Otherwise, we call the query \emph{long}.
\end{definition}

We consider the following loose bounds: if there are \(k\) candidates for vertices that could be the sink (i.e. \(\textsf{high} - \textsf{low} = k\)), then a long query requires the exploration of at least \(\omega(1)\) vertices and by performing it, it eliminates at most \(n/2\) vertices from the list of candidates to be the sink; a short query requires the exploration of at least \(\Omega(1)\) vertices, and, in the worst case, it eliminates at most \(h\) vertices from the list of candidates.

We show that no matter the order in which the algorithm performs long and short queries, it cannot guarantee to find the sink before exploring at least \(\omega(\log n)\) vertices.

\begin{lemma}\label{lem:alg_bound_long_queries}
    If Algorithm~\ref{alg:finding_sink} performs only long queries and explores \(O(\log n)\) vertices, then it can eliminate at most \(O(n^{1/4})\) of the candidates to be the sink.
\end{lemma}
\begin{proof}
    Considering that each long query requires the exploration of \(\omega(1)\) vertices, the algorithm cannot perform more than \(o(\log n)\) of these long queries before it exceeds the \(O(\log n)\) limit placed on the total number of explored vertices. Write \((\log n)/f(n)\) for the amount of long queries performed, with \(f(n) \to +\infty\) when \(n\to+\infty\). These long queries reduce the number of candidates by a factor at most \(2^{(\log n)/f(n)} = n^{1/f(n)}\). It follows that after all long queries have been performed there still remain \(\Omega(n^{1 - 1/f(n)})\) candidates for the sink, which can be bounded from below by \(\Omega(n^{\alpha})\) for any \(0 < \alpha < 1\). For example, picking \(\alpha = 3/4\) provides the statement from Lemma~\ref{lem:alg_bound_long_queries}.
\end{proof}

In the following, we will show a lower bound on the number of short queries:

\begin{theorem}\label{thm:alg_lb_queries}
    Algorithm~\ref{alg:finding_sink} must explore at least \(\omega(\log n)\) vertices in order to find the sink.
\end{theorem}
\begin{proof}
After performing a certain number of long queries as analyzed in Lemma~\ref{lem:alg_bound_long_queries}, the algorithm may no longer use long queries, and it only has access to short queries to find the sink among \(n^{3/4}\) candidates.
Assume a worst-case execution: each time the algorithm performs a query, the sink is revealed to be in the biggest of the two halves -- and halt this execution when there are \(n^{1/4}\) remaining candidates. To get to this point, there must have been at least \[\dfrac{n^{3/4} - n^{1/4}}{\log N}\] queries performed. Observe further that the probability that such a worst-case occurs is the probability that the actual sink is among the remaining \(n^{1/4}\) vertices.
This worst-case execution alone causes at least \[\dfrac{n^{3/4} - n^{1/4}}{\log n}\dfrac{n^{1/4}}{n^{3/4}} = \dfrac{n^{1/4}}{\log n}(1 + o(1)) = \omega(\log n)\] vertices to be explored on average. The total cost on average is at least as high as the weighted cost attributed solely to a few worst-case instances. Hence the entire algorithm cannot find the sink without exploring at least \(\omega(\log n)\) vertices on average.
\end{proof}

Thus Theorem~\ref{thm:alg_lb_queries} shows that no algorithm can find a mend for Problem~\ref{prob:path_to_sink} with expected volume \(O(\log n)\). Since we have shown in Section~\ref{sec:minvolume_of_path_to_sink} that the \(\minvolume\) is \(\Theta(\log n)\), this shows a separation between \(\minvolume\) and \(\avgvolume\). Sections~\ref{sec:path_to_sink_general_graphs} and~\ref{sec:unorient} show that unlike the separation shown in Section~\ref{sec:separation_min_avg}, this one does not need to rely on a promise.

This proves that there is a separation between \(\minvolume\) and \(\avgvolume\), since we have constructed a problem that one can solve by modifying \(O(\log n)\) labels with a complete view of the graph and yet requires \(\omega(\log n)\) queries when one only has access to randomness.

Note that many of the bounds we have used here are very loose, most notably the actual cost of a long query is estimated to be \(\Omega(\log \log N)\) rather than \(\omega(1)\). In practice, the average number of vertices visited seems to be closer to \(\Omega(\log^2 N)\), even assuming that the tree has the correct structure. Further, it is not clear whether a randomized algorithm could even achieve such a bound if it also had to take into account inputs that do not satisfy the local constraints of \(\tree\).

\subsection{Extension to unoriented graphs}\label{sec:unorient}

The results from Sections~\ref{sec:minvolume_of_path_to_sink} and~\ref{sec:avgvolume_of_path_to_sink} extend to unoriented graphs simply by encoding the orientation within the structure of the graph. One possible encoding of the orientation is the following: given an unoriented graph \(G\), let \emph{leaves} be the vertices with degree exactly $1$. To interpret \(G\) as an oriented graph \(\overline{G}\) with maximum degree \(\Delta\), extract the subset \(\overline{V}\) of vertices that have at least \(\Delta+1\) leaves as neighbors. This provides the vertex set of \(\overline{G}\). An edge between \(\overline{u}\) and \(\overline{v}\) is added if between the corresponding vertices \(u\) and \(v\) in \(G\), there is a path of length exactly 2, connected to at most one leaf. If there is such a leaf on the vertex of the path closest to \(u\), consider the edge oriented from \(\overline{v}\) to \(\overline{u}\). Otherwise, if there is no leaf then the edge is not oriented. Figure~\ref{fig:unorient} visualizes this transformation.

A labeling \(\lambda\) of \(G\) is interpreted as a labeling \(\overline{\lambda}\) of \(\overline{G}\) by defining \(\overline{\lambda}(\overline{v}) \coloneqq \lambda(v)\). Any local property of \(\overline{\lambda}\) can be computed locally on \(\lambda\). Problem~\ref{prob:path_to_sink_unoriented} shows how this encoding can be used to define a version of Problem~\ref{prob:path_to_sink} that works on unoriented graphs.

\begin{problemdef}
    \caption{Path to the sink -- unoriented graphs\problemcaptionkludge}
    \label{prob:path_to_sink_unoriented}
    \vspace{-\topsep}\begin{description}[noitemsep]
        \item[Input:] An unoriented graph \(G\)
        \item[Labels:] \textsf{red} and \textsf{black} (same as in Problem~\ref{prob:path_to_sink} and Problem~\ref{prob:path_to_sink_oriented})
        \item[Task:] Label \(G\) by \(\lambda\) so that \(\overline{\lambda}\) is a labeling of \(\overline{G}\) that satisfies the constraints of Problem~\ref{prob:path_to_sink_oriented}
    \end{description}\vspace{-1.1\topsep}
\end{problemdef}

\begin{figure}[htb]
    \centering\includegraphics[width=10cm]{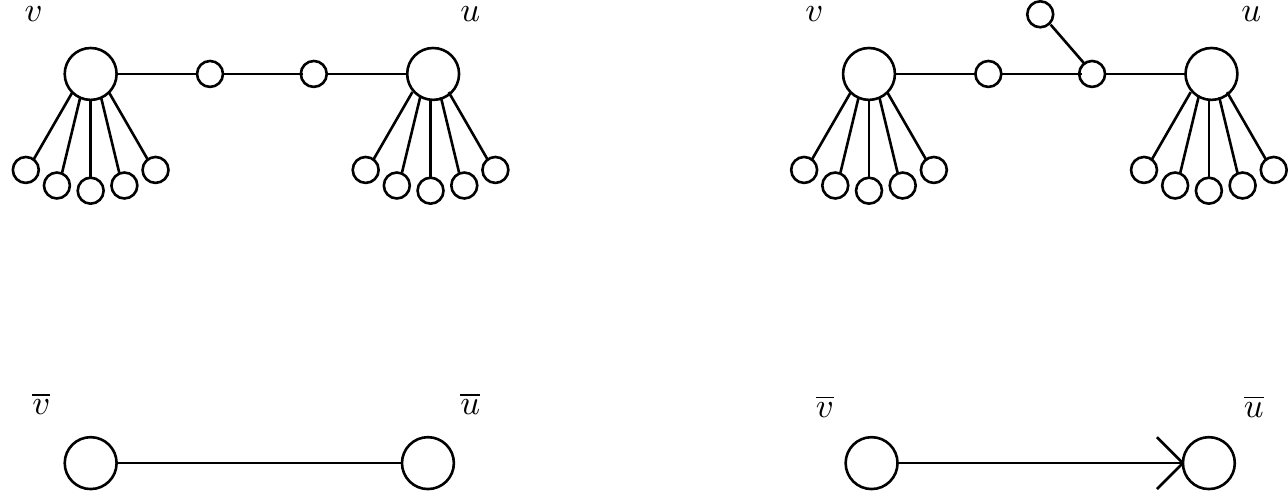}
    \caption{Oriented interpretation of an unoriented graph. Top: pairs of vertices in \(G\); bottom: their interpretation in \(\overline{G}\). This mapping can be computed locally.}
    \label{fig:unorient}
\end{figure}

\section{Landscape of mending volume}
\label{sec:landscape_appendix}

Table~\ref{tab:complexities_overview} provides an overview of the landscape only for the existential mending volume. We can perform the same work for the other measures that we have defined, which provides the results summarized in Tables~\ref{tab:complexities_complete_1}-\ref{tab:complexities_complete_3}. The results for the mending radius follow immediately from Lemma~\ref{lem:radius_equivalence} which makes the results from~\cite{local_mending} applicable to \(\radius\). 

\begin{table}[tb]
    \newcommand{\ok}{{\textcolor{blue!80!black}{\checkmark}}}
    \newcommand{\ko}{{\textcolor{red!80!black}{\textbf{\texttimes}}}}
    \newcommand{\uk}{{\textcolor{black}{?}}}
    \caption{An overview of the landscape of the existential mending volume ($\minvolume$), the expected mending volume (\(\avgvolume\)), the deterministic mending volume (\(\maxvolume\)), and the mending radius (\(\radius\)) for LCL problems on the classes of paths, trees and general graphs. Here \ok{} denotes that LCL problems with this mending volume exist, \ko{} denotes that such LCL problems cannot exist, and \uk{} denotes an open question. We assume $k>1$ and $0<\alpha<1$.}
    \label{tab:complexities_complete}
    \begin{subtable}{1\textwidth}
        \caption{Landscape of \(\minvolume\) and \(\avgvolume\)}
        \label{tab:complexities_complete_1}
        \centering
        \begin{tabular}{lccccccc}
            \toprule
            Setting & \multicolumn{7}{l}{Possible existential and expected mending volumes} \\ \cmidrule{2-8}
        & \(O(1)\) & \ldots & \(\Theta(\log n)\) &  \(\Theta(\log^k n)\) & \ldots & \(\Theta(n^{\alpha})\) & \(\Theta(n)\) \\
            \midrule
            Paths and cycles & \ok  & \ko  & \ko  & \ko  & \ko  & \ko  & \ok \\
            Rooted trees     & \ok  & \ko  & \ok  & \ok  & \ko  & \ok  & \ok \\
            Trees            & \ok  & \ko  & \ok  & \ok  & \uk  & \ok  & \ok \\
            General graphs   & \ok  & \ko  & \ok  & \ok  & \uk  & \ok  & \ok \\
            \bottomrule
        \end{tabular}
    \end{subtable}
    \newline
    \vspace*{7pt}
    \newline  
    \begin{subtable}{1\textwidth}
        \caption{Landscape of \(\maxvolume\)}
        \label{tab:complexities_complete_2}
        \centering
        \begin{tabular}{lccccccc}
            \toprule
            Setting & \multicolumn{7}{l}{Possible deterministic mending volumes} \\ \cmidrule{2-8}
            & \(O(1)\) & \ldots & \(\Theta(\log n)\) & \(\Theta(\log^k n)\) & \ldots & \(\Theta(n^{\alpha})\) & \(\Theta(n)\) \\
            \midrule
            Paths and cycles & \ok  & \ko  & \ko  & \ko  & \ko  & \ko  & \ok \\
            Rooted trees     & \ok  & \ko  & \ko  & \ko  & \ko  & \ko  & \ok \\
            Trees            & \ok  & \ko  & \ko  & \ko  & \ko  & \ko  & \ok \\
            General graphs   & \ok  & \ko  & \ko  & \ko  & \ko  & \ko  & \ok \\
            \bottomrule
        \end{tabular}
    \end{subtable}
    \newline
    \vspace*{7pt}
    \newline        
    \begin{subtable}{1\textwidth}
        \caption{Landscape of \(\radius\)}
        \label{tab:complexities_complete_3}
        \centering
        \begin{tabular}{lccccccc}
            \toprule
            Setting & \multicolumn{7}{l}{Possible mending radius} \\ \cmidrule{2-8}
            & \(O(1)\) & \ldots & \(\Theta(\log n)\) & \(\Theta(\log^k n)\) & \ldots & \(\Theta(n^{\alpha})\) & \(\Theta(n)\) \\
            \midrule
            Paths and cycles & \ok  & \ko  & \ko  & \ko  & \ko  & \ko  & \ok \\
            Rooted trees     & \ok  & \ko  & \ok  & \ko  & \ko  & \ko  & \ok \\
            Trees            & \ok  & \ko  & \ok  & \ko  & \ko  & \ko  & \ok \\
            General graphs   & \ok  & \ko  & \ok  & \uk  & \uk  & \ok  & \ok \\
            \bottomrule
        \end{tabular}
    \end{subtable}
\end{table}

\subsection{Landscape of \avgvolumetitle}\label{sec:landscape_avgvolume}

The proof of Theorem~\ref{thm:unbalance} shows that for propagation problems, the expected mending volume is equivalent to the existential mending volume. The same constructions apply to encode complexity results obtained in infinite rooted trees to general trees and general graphs.

Although we have provided a problem for which \(\minvolume = o(\avgvolume)\) in Section~\ref{sec:separation_min_avg_general}, this does not imply that the possible complexities are different, merely that they are obtained for different problems. We conjecture that Problem~\ref{prob:path_to_sink} has \(\avgvolume = \Theta(\log^2 n)\), which is a complexity that exists for the \(\minvolume\) on other problems.

Thus to the best of our knowledge, \(\avgvolume\) exhibits the same possible complexities as \(\minvolume\).

\subsection{Landscape of \maxvolumetitle}\label{sec:landscape_maxvolume}

The \(\maxvolume\) from Definition~\ref{def:complexity_generalized} can be shown to be a deterministic volume. Indeed if all transition probabilities are \(\mender_{G,\lambda,v}(W,W') \in \set{0,1}\) then so is \(\mender^*_{G,\lambda,v}(\set{v},W) \in \set{0,1}\). The Markov chain collapses to a deterministic computational process that decides which vertices to explore next.

Our goal is to show that \(\maxvolume\) is either \(O(1)\) or \(\Theta(n)\) on most classes of graphs. This is stated by Theorem~\ref{thm:maxvolume_trivial}.

We first prove that knowledge of the size of the instance does not increase the computational power. This means that the definition is equivalent to its variant where the next step chosen by \(\mender\) can depend on the size of the input graph, \(n\). This invariance under the size of the instance is essential for Theorem~\ref{thm:maxvolume_trivial}.

\begin{lemma}[Knowledge of the size of the instance is useless]\label{size_useless}
    The version \(\maxvolume'\) of the deterministic volume where \(\mender_{G,\lambda,v}\) can depend on \(n\) the size of the instance is asymptotically equivalent to \(\maxvolume\).
\end{lemma}
\begin{proof}
    This proof uses a similar technique to one featured in~\cite{removing_global_assumptions}: if an upper bound on the size of the instance is necessary for the execution of the algorithm, we can simulate the algorithm by guessing increasingly large values of the size of the instance until the algorithm executes correctly.

    Assume that a problem has \(\maxvolume' = \Theta(g(n))\) where \(n\) is known to the mender \(\mender\). Assume further that \(g\) is upper bounded by some known \(f\), where \(f\) is computable and increasing function. Let \(f^{-1}(y)\) be the single value \(x\) such that \(f(x-1) < y \leq f(x)\). Note that this is computable as well.
    Algorithm~\ref{alg:simulation} presents a strategy that computes a mend of \(\lambda\) at \(x\) without the knowledge of \(n\), by guessing the size of the graph by successively increasing it.

    \begin{algorithm}[tb]
        \begin{algorithmic}[1]
            \State \(m \gets 1\)
            \Repeat
                \State simulate an algorithm \(A\) computing \(\maxvolume'\) assuming size \(m\)
                \If{\(\maxvolume'\) explores more than \(f(m) + 1\) vertices}
                    \State abort the simulation
                    \State \(m \gets f^{-1}(2 \cdot f(m))\)
                \EndIf
            \Until{\(A\) halts naturally}
        \end{algorithmic}
        \caption{A simulation without knowledge of the size of the graph}
        \label{alg:simulation}
    \end{algorithm}

    The main loop from Algorithm~\ref{alg:simulation} must terminate since, eventually, \(m \geq n\), and a mend is found for the first such \(m\). Let \(m_0,m_1,\dotsc,m_k\) be the successive values of \(m\) at each iteration. Observe that \(f(m_{i+1}) \geq 2\cdot f(m_i) + 1\). Thus at each iteration the new explored vertices are as many as all previously visited vertices, so the total number of vertices explored does not exceed
    \begin{align*}
        2 \times f(m_k)
        &\leq 4 \cdot f(m_{k-1}) &&\triangleright\text{ by definition of \(m_i\)} \\
        &\leq 4 \cdot f(n). &&\triangleright\text{ since \(f\) is increasing}
    \end{align*}
    Thus, only at a cost of a constant multiplicative factor, we can execute an algorithm that requires the knowledge of \(n\) in a context where \(n\) is unknown. Or equivalently, where the result must be independent of \(n\). This means, that the mend computed by \(\maxvolume'\) is independent of the size of the graph and thus equivalent to \(\maxvolume\).
\end{proof}

We can now establish a complete landscape of the deterministic mending volume in Theorem~\ref{thm:maxvolume_trivial}.
\begin{theorem}[\(\maxvolume\) is trivial]\label{thm:maxvolume_trivial}
    For any problem on paths, trees, or general graphs, \(\maxvolume\) is either \(\Theta(1)\) or \(\Theta(n)\).
\end{theorem}
\begin{proof}
    Assume that \(\maxvolume = \omega(1)\), then there is a sequence \((G_i)_{i\in\NN}\) of graphs for which there are eventually arbitrarily many labels that are modified by the deterministic mender:
    \(|M_i| \to \infty\) as \(i \to \infty\), where \(M_i\) is the set of vertices of \(G_i\) whose label is modified.
    Define \(N_i\) to be the radius-\(r\) neighborhood of \(M_i\), and note that the size of \(N_i\) is at most \(|M_i| \times \Delta^r\). Since \(N_i\) is indistinguishable from \(G_i\) on all radius-\(r\) neighborhoods of elements of \(M_i\), the same mending process operating on \(N_i\) will also explore exactly \(M_i\). Hence for the sequence \((N_i)_{i\in\NN}\) of graphs of arbitrary size, the mend is of size at least \[|M_i| \geq \dfrac{|N_i|}{\Delta^r} = \Omega(n).\]

    This proves a gap in deterministic volume: on the usual classes of paths, cycles, trees, and general graphs, there exists no problem with \(\maxvolume\) between \(\omega(1)\) and \(o(n)\).
\end{proof}

\begin{corollary}[A deterministic mender explores the entire neighborhood]
    \(\maxvolume\) is asymptotically equivalent to \(|\neigh_{\radius}|\).
\end{corollary}
\begin{proof}
    From~\cite{local_mending} and the equivalence established in Lemma~\ref{lem:radius_equivalence}, we know that \(\radius\) can only be one of \(\Theta(1)\) or \(\Theta(n)\) on paths and cycles, and \(\Theta(1)\) or \(\Theta(\log n)\) or \(\Theta(n)\) on trees.

    If \(\radius\) is \(\Theta(1)\) then \(\maxvolume\) is \(\Theta(1)\) as well. If \(\radius\) is \(\Theta(\log n)\) or \(\Theta(n)\) then \(\maxvolume\) is \(\Theta(n)\). In paths, cycles, and trees, these are also the sizes of the neighborhoods of radius \(\radius\).
\end{proof}

\end{document}